\documentclass[11pt]{article}

\usepackage[margin=1in]{geometry}
\usepackage[T1]{fontenc}
\IfFileExists{lmodern.sty}{\usepackage{lmodern}}{}
\usepackage{amsmath,amssymb,amsthm,mathtools}
\usepackage{graphicx}
\usepackage{bm}
\usepackage{booktabs}
\usepackage{array}
\usepackage{enumitem}
\usepackage{aliascnt}
\usepackage[numbers,sort&compress]{natbib}
\usepackage[colorlinks=true,allcolors=blue]{hyperref}
\usepackage[nameinlink,capitalise,noabbrev]{cleveref}
\usepackage{float}
\usepackage{graphicx,tikz}

\allowdisplaybreaks
\setlist[itemize]{leftmargin=1.8em,itemsep=0.2em,topsep=0.3em}
\setlist[enumerate]{leftmargin=2.0em,itemsep=0.2em,topsep=0.3em}

\newtheorem{theorem}{Theorem}[section]
\newaliascnt{proposition}{theorem}
\newtheorem{proposition}[proposition]{Proposition}
\aliascntresetthe{proposition}
\newaliascnt{lemma}{theorem}
\newtheorem{lemma}[lemma]{Lemma}
\aliascntresetthe{lemma}
\newaliascnt{corollary}{theorem}
\newtheorem{corollary}[corollary]{Corollary}
\aliascntresetthe{corollary}
\newaliascnt{fact}{theorem}

\aliascntresetthe{fact}
\theoremstyle{definition}
\newaliascnt{definition}{theorem}
\newtheorem{definition}[definition]{Definition}
\aliascntresetthe{definition}
\newaliascnt{problem}{theorem}

\aliascntresetthe{problem}
\theoremstyle{remark}
\newaliascnt{remark}{theorem}
\newtheorem{remark}[remark]{Remark}
\aliascntresetthe{remark}

\newcommand{\C}{\mathbb{C}}
\newcommand{\R}{\mathbb{R}}

\newcommand{\Id}{\mathrm{id}}
\newcommand{\Tr}{\operatorname{Tr}}

\newcommand{\poly}{\operatorname{poly}}
\newcommand{\Herm}{\operatorname{Herm}}
\newcommand{\ip}[2]{\left\langle #1,#2\right\rangle}
\newcommand{\norm}[1]{\left\lVert #1\right\rVert}
\newcommand{\abs}[1]{\left\lvert #1\right\rvert}
\newcommand{\ket}[1]{\lvert #1\rangle}
\newcommand{\bra}[1]{\langle #1\rvert}
\newcommand{\proj}[1]{\lvert #1\rangle\!\langle #1\rvert}
\newcommand{\Psucc}{P_{\mathrm{succ}}}
\newcommand{\etatr}{\eta_{\mathrm{tr}}}
\newcommand{\sphere}{\mathbb{S}}
\newcommand{\E}{\mathbb{E}}
\newcommand{\cI}{\mathcal{I}}
\newcommand{\cK}{\mathcal{K}}
\newcommand{\cX}{\mathcal{X}}
\newcommand{\cH}{\mathcal{H}}
\newcommand{\cD}{\mathcal{D}}
\newcommand{\cL}{\mathcal{L}}
\newcommand{\cY}{\mathcal{Y}}
\newcommand{\bdeg}{\operatorname{bdeg}}

\newcommand{\lmin}{\lambda_{\min}}
\newcommand{\lmax}{\lambda_{\max}}
\newcommand{\Pharm}{\mathcal{P}^{\mathrm{harm}}}

\title{A Sum-of-Squares Hierarchy with Quadratic Convergence for Quantum Channel Coding}

\author{
  Hoang Ta\\
  {\small Hanoi University of Science and Technology}\\
  {\small Vietnam}
  \and
  Hoang Anh Tran\\
  {\small National University of Singapore}\\
  {\small Singapore}
}

\date{}

\begin{document}
\maketitle

\begin{abstract}
Computing the optimal success probability for transmitting classical messages through a single use of a quantum channel is NP-hard, even for two messages. An existing semidefinite programming hierarchy based on symmetric extensions provides convergent upper bounds with an a priori error estimate that decays as the inverse square root of the extension level. In this work, we construct a Hermitian sum-of-squares hierarchy for an arbitrary number of messages and prove quadratic convergence in its level. The error bound is proportional to the advantage over random guessing. Our approach combines state-discrimination duality with positive polynomial kernels on products of spheres to construct feasible polynomial dual certificates. For binary messages, the resulting bounds give a multiplicative approximation from above of the trace-norm contraction coefficient.
\end{abstract}

%%%%%%%%%%%%%%%%%
\section{Introduction}
\label{sec:intro}

One of the main tasks in information theory is to determine how reliably a message can be transmitted through a noisy channel. For classical channels, Shannon's channel coding theorem characterizes the largest communication rate achievable with vanishing error in the asymptotic limit of many independent channel uses \cite{Shannon1948}. In the one-shot setting, the channel is used only once, and the objective is to maximize the probability of correctly transmitting one of a prescribed number of messages over all admissible encoders and decoders.

From an algorithmic perspective, the one-shot coding problem takes an explicit description of a classical channel and a prescribed number of messages as input, and asks for an encoder and decoder that maximize the success probability. This optimization can be formulated as submodular maximization, leading to a polynomial-time greedy algorithm with approximation ratio $1-e^{-1}$ \cite{BarmanFawzi2018}. However, achieving a strictly larger constant approximation ratio is NP-hard when the number of messages is part of the input.

In this work, we consider the transmission of classical information through a quantum channel. Let $A$ and $B$ be finite-dimensional complex Hilbert spaces, with $d_A=\dim A$ and $d_B=\dim B$, and write $\cL(A)$ for the space of linear operators on $A$. A quantum channel $\Phi:\cL(A)\to\cL(B)$ is a completely positive, trace-preserving linear map. For an integer $k\ge2$, let $\Psucc(\Phi,k)$ denote the largest average probability of correctly transmitting one of $k$ equiprobable classical messages through a single use of $\Phi$. The sender encodes the messages into quantum states, and the receiver applies a positive operator-valued measure (POVM) to distinguish their channel outputs. Both the input states and the decoding measurement are optimized, without shared entanglement.

Since classical channels are special cases of quantum channels, this problem inherits the classical hardness results. In the quantum setting, however, computing the optimal success probability remains NP-hard even for two messages~\cite{DelsolEtAl2025}, whereas classical binary coding is efficiently solvable. More precisely,
\[
 \Psucc(\Phi,2)=\frac{1+\etatr(\Phi)}2,
\]
where $\etatr(\Phi)$ is the trace-norm contraction coefficient, defined as the largest ratio of output trace distance to input trace distance over distinct input states, with the convention $\etatr(\Phi)=0$ when $d_A=1$. Approximating this coefficient within a constant factor and deciding whether it equals one are NP-hard \cite{DelsolEtAl2025}. Binary coding therefore connects the computational problem directly with the loss of distinguishability under a quantum channel.

Several approaches provide bounds and approximation guarantees in quantum communication settings. Approximation algorithms have been obtained for classical--quantum channels \cite{FawziSeifSzilagyi2019} and entanglement-assisted coding \cite{OufkirBerta2025}, the latter involving a different feasible set from the unassisted problem considered here. Hypothesis testing and non-signaling or PPT-preserving relaxations yield converse bounds for classical communication over quantum channels \cite{WangRenner2012,MatthewsWehner2014,WangXieDuan2018}. On the contraction side, connections with channel orders and capacity approximation are studied in \cite{HircheRouzeFranca2022}, while quantum Doeblin coefficients provide efficiently computable bounds \cite{Hirche2024,GeorgeEtAl2026}.

For the optimal success probability $\Psucc(\Phi,k)$, a convergent semidefinite programming (SDP) hierarchy is constructed in \cite{DelsolEtAl2025} using the constrained bilinear framework of \cite{BertaEtAl2021}. For every integer $m\ge1$, its level-$m$ value satisfies
\[
 0\le\mathrm{SDP}_m(\Phi,k)-\Psucc(\Phi,k)\le\frac{\poly(d)}{\sqrt m},\qquad d=\max\{d_A,d_B\},
\]
where $\poly(d)$ denotes a polynomial prefactor depending on the channel dimensions. The hierarchy provides upper bounds that can be strengthened by positive partial transpose (PPT) constraints. Its direct formulation involves matrices acting on an input system and $m$ output copies, making the level required for a prescribed accuracy an important factor in the computational cost. This motivates the construction of hierarchies with stronger quantitative convergence guarantees.

\paragraph{Our approach and main results.}
We construct a Hermitian sum-of-squares (SOS) hierarchy for $\Psucc(\Phi,k)$. The construction combines a polynomial parametrization of pure input states with state-discrimination duality, and replaces the resulting pointwise positivity constraints by finite-degree SOS certificates. For every integer $L\ge1$, the level-$L$ value $U_L(\Phi,k)$ is the optimal value of a finite-dimensional SDP. These values form a nonincreasing sequence of upper bounds on $\Psucc(\Phi,k)$.

Our main result establishes quadratic convergence with an error proportional to the advantage over random guessing. There are absolute constant $c_1>0$ such that, for every $L\ge 1$,
\begin{equation}
 0\le U_L(\Phi,k)-\Psucc(\Phi,k)\le4c_1\frac{d_A^2}{L^2}\left(\Psucc(\Phi,k)-\frac1k\right).
\label{eq:intro-main-rate}
\end{equation}
Consequently, the choice $L=O(d_A/\sqrt{\varepsilon})$ guarantees an additive relaxation error of at most $\varepsilon\in(0,1)$. This sufficient level is independent of the output dimension $d_B$, although the SDP size depends on it. The dependence on $\Psucc(\Phi,k)-1/k$ also gives a stronger guarantee for channels whose coding performance is close to random guessing.

For binary messages, this relative error bound becomes a multiplicative approximation of the trace-norm contraction coefficient. Defining $E_L(\Phi):=2U_L(\Phi,2)-1$, we obtain
\begin{equation}
 \etatr(\Phi)\le E_L(\Phi)\le\left(1+4c_1\frac{d_A^2}{L^2}\right)\etatr(\Phi)\qquad(L\ge 1).
\label{eq:intro-binary-rate}
\end{equation}
Moreover, every level $L\ge1$ is exact when $\etatr(\Phi)\in\{0,1\}$. We compare the SDP sizes of the SOS and extension hierarchies at a common guaranteed accuracy. Numerical experiments on $40$ sampled qubit-to-qutrit channels further indicate that the first SOS level is numerically tight throughout the sample and improves on the first extension level, both with and without PPT constraints.

\subsection{Technical overview}
\label{sec:intro-technical}

\emph{From channel coding to polynomial certificates.}
For a fixed decoding POVM, each input state can be chosen pure without decreasing the success probability. Parametrizing these pure states by real unit vectors gives an encoder parameter
$x=(x_1,\ldots,x_k)\in\cX=(\sphere^{2d_A-1})^k$.
Let $\sigma_i(x)$ denote the channel output associated with the input represented by $x_i$. State-discrimination duality eliminates the decoding measurement and gives
\[
 s(x)=\min\left\{\Tr Y:\ Y\in\Herm(B),\
 Y\succeq\frac{\sigma_i(x)}k\text{ for all }i\in[k]\right\},
 \qquad
 P:=\Psucc(\Phi,k)=\max_{x\in\cX}s(x).
\]
Thus an upper bound $\gamma$ is certified by a Hermitian matrix field $Y(x)$ satisfying
$Y(x)\succeq\sigma_i(x)/k$ and $\Tr Y(x)\le\gamma$ throughout $\cX$.
At level $L$, we restrict $Y$ to ambient Hermitian matrix polynomials of degree at most $2L$ in each input block and impose matrix and scalar SOS certificates modulo the sphere equations, with factors of degree at most $L$ in each block. Their finite Gram representations give the SDP defining $U_L(\Phi,k)$.

\emph{Turning positive fields into SOS certificates.}
The main obstacle in proving convergence is that a pointwise optimal dual field need not be polynomial. For example, in binary discrimination, a natural optimal solution involves the matrix absolute value
$\abs{\sigma_1(x)-\sigma_2(x)}$.
We choose a bounded Borel measurable optimal field $Y_*$ with
$\Tr Y_*(x)=s(x)$ and apply a product squared-polynomial kernel $\cK_q$ directly to the positive matrix slacks $Y_*-\sigma_i/k$ and the nonnegative scalar slack $P-s$.
If the real polynomial $q$ has degree at most $L$, the kernel maps these fields to ambient polynomials of degree at most $2L$ in each block. Integrating their Gram coefficients yields finite positive semidefinite Gram matrices, so the resulting polynomials admit exact SOS representations. This construction requires no polynomial approximation or continuity of $Y_*$. However, the kernel also changes the output fields $\sigma_i$, and its effect on the constraints must be corrected.

\emph{Correcting the constraints through their common harmonic structure.}
All output fields have the same spherical mean
$\tau=\Phi(I_A/d_A)$ and satisfy
$\sigma_i(x)=\tau+h(x_i)$, where $h$ is a Hermitian matrix-valued spherical harmonic of degree two.
After normalization, $\cK_q$ fixes constant fields on $\cX$ and multiplies each $h(x_i)$ by the same eigenvalue $\lambda=\lambda_2(q)$.
For $\lambda>0$, define
\[
 \hat{Y}_L
 =\lambda^{-1}\cK_qY_*-\frac{1-\lambda}{\lambda k}\tau,
 \qquad
 \hat{\gamma}_L
 =\frac{P}{\lambda}-\frac{1-\lambda}{\lambda k}.
\]
On $\cX$, the corrected matrix and scalar slacks equal
$\lambda^{-1}\cK_q(Y_*-\sigma_i/k)$ and
$\lambda^{-1}\cK_q(P-s)$, respectively.
Their SOS representations therefore establish feasibility at level $L$, with objective
\[
 \hat{\gamma}_L
 =P+(\lambda^{-1}-1)\left(P-\frac1k\right).
\]
The relative factor $P-1/k$ comes from the common baseline $\tau/k$ in the affine correction. Choosing $q$ to attain $\lambda_{2d_A,L}$ and applying \cref{prop:kernel-constant} gives $\lambda^{-1}-1\le4c_1d_A^2/L^2$ for every $L\ge1$, yielding \eqref{eq:intro-main-rate}. \Cref{fig:intro-overview} summarizes the construction and the convergence argument. 
% Choosing $q$ using the spherical kernel estimates of \cite{FangFawzi2021} gives
% $\lambda^{-1}-1\le4c_1d_A^2/L^2$ for $L\ge 1$, yielding \eqref{eq:intro-main-rate}.

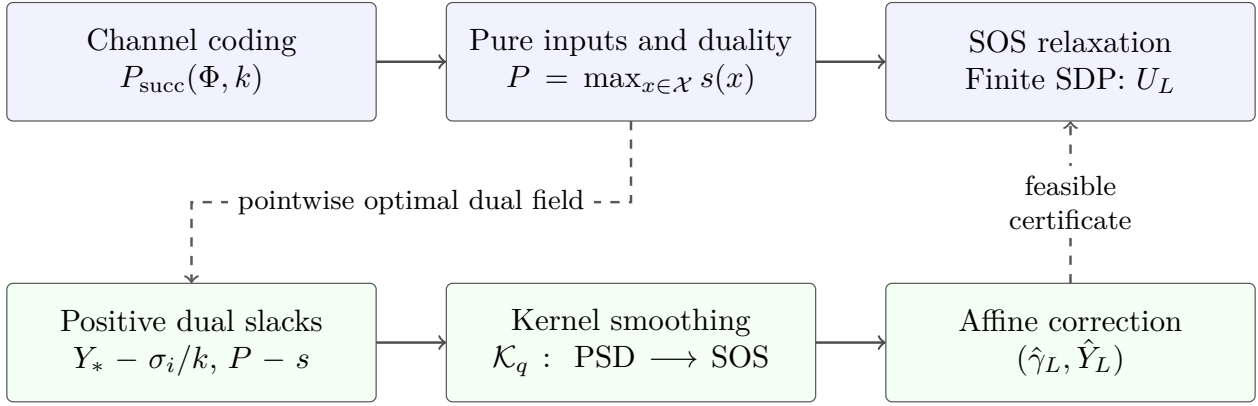
\begin{figure}[t]
\centering
\resizebox{\linewidth}{!}{%
\begin{tikzpicture}[
  font=\small,
  sosbox/.style={
    draw=black!65,
    rounded corners=2pt,
    text width=3.85cm,
    minimum height=1.35cm,
    align=center,
    inner sep=5pt,
    outer sep=0pt
  },
  sosmethod/.style={sosbox,fill=blue!5},
  sosproof/.style={sosbox,fill=green!5},
  sosflow/.style={->,thick,draw=black!70},
  soslink/.style={->,thick,dashed,draw=black!65}
]
\node[sosmethod] (coding) at (0,0)
  {Channel coding\\$\Psucc(\Phi,k)$};
\node[sosmethod] (dual) at (5,0)
  {Pure inputs and duality\\$P=\max_{x\in\cX}s(x)$};
\node[sosmethod] (sdp) at (10,0)
  {SOS relaxation\\Finite SDP: $U_L$};

\node[sosproof] (slacks) at (0,-3.2)
  {Positive dual slacks\\$Y_*-\sigma_i/k$, $P-s$};
\node[sosproof] (kernel) at (5,-3.2)
  {Kernel smoothing\\$\cK_q:\ \mathrm{PSD}\longrightarrow\mathrm{SOS}$};
\node[sosproof] (correction) at (10,-3.2)
  {Affine correction\\$(\hat{\gamma}_L,\hat{Y}_L)$};

\draw[sosflow] (coding.east) -- (dual.west);
\draw[sosflow] (dual.east) -- (sdp.west);
\draw[sosflow] (slacks.east) -- (kernel.west);
\draw[sosflow] (kernel.east) -- (correction.west);

\draw[soslink]
  (dual.south) -- (5,-1.6) -- (0,-1.6) -- (slacks.north);
\node[font=\footnotesize,fill=white,inner sep=3pt]
  at (2.5,-1.6) {pointwise optimal dual field};

\draw[soslink] (correction.north) -- (sdp.south);
\node[font=\footnotesize,align=center,fill=white,inner sep=3pt]
  at (10,-1.6) {feasible\\certificate};
\end{tikzpicture}%
}
\caption{Construction of the SOS hierarchy (upper row) and feasible certificates for its convergence analysis (lower row). The kernel is used in the proof; solving the finite SDP does not require knowing the optimal field $Y_*$ or the value $P$.}
\label{fig:intro-overview}
\end{figure}

\subsection{Related work}
\label{sec:intro-related}

\emph{Channel coding and contraction coefficients.}
 Related algorithmic work considers broadcast channels \cite{fawzi2024broadcast}, multiple-access channels with non-signaling assistance \cite{FawziFerme2024MAC}, and classical--quantum channels \cite{FawziSeifSzilagyi2019}. For classical channels, non-signaling codes also yield linear programming converse bounds \cite{Matthews2012}. For classical communication over quantum channels, converse bounds have been developed through hypothesis testing \cite{WangRenner2012,MatthewsWehner2014}, non-signaling and PPT-preserving relaxations \cite{WangXieDuan2018}, meta-converses \cite{WangFangTomamichel2019}, and extendible measurements under a maximum-error criterion \cite{SinghNuradhaWilde2025}. Entanglement-assisted coding admits approximation algorithms \cite{OufkirBerta2025} and position-based achievability bounds \cite{QiWangWilde2018}, with shared entanglement enlarging the set of admissible encoding and decoding strategies. The binary case of our problem connects channel coding with trace-norm contraction, whose structural properties were studied in \cite{Ruskai1994}. Subsequent work relates contraction coefficients to channel orders and capacity approximation \cite{HircheRouzeFranca2022}, while quantum Doeblin coefficients provide efficiently computable bounds \cite{Hirche2024,GeorgeEtAl2026}. Other developments concern quantum local differential privacy \cite{NuradhaWilde2025}, nonlinear strong data-processing inequalities for quantum hockey-stick divergences \cite{NuradhaGeorgeHirche2025}, and conditional contraction coefficients with quantum reference systems \cite{HircheEtAl2026}.

\emph{Semidefinite hierarchies in quantum information.}
Moment relaxations and SOS certificates provide a general framework for polynomial optimization \cite{Lasserre2001,Parrilo2003,Laurent2009}. In quantum information, symmetric extensions approximate separability \cite{DohertyParriloSpedalieri2004}, while noncommutative moment hierarchies address quantum correlations and bilinear optimization \cite{NavascuesPironioAcin2008,PironioNavascuesAcin2010,BertaFawziScholz2016}. SDP bounds for quantum communication and error correction also arise from PPT-preserving and non-signaling codes \cite{LeungMatthews2015} and extension-based methods \cite{BertaEtAl2021,KaurEtAl2021,HoldsworthSinghWilde2023}. Quantitative results include quadratic convergence for PPT Bose-symmetric extensions \cite{NavascuesOwariPlenio2009}. Exploiting symmetry reduces the size of SDPs for channel fidelity \cite{CheeTaVu2025} and regularized channel quantities \cite{FawziShayeghiTa2022}; see \cite{TavakoliEtAl2024} for a broader review. Of particular relevance to our analysis, squared-polynomial kernels yield quadratic convergence of SOS hierarchies on the sphere, with relative error bounds and extensions to matrix-valued polynomials \cite{FangFawzi2021}. The kernel approach also gives quadratic convergence on products of spheres, with an application to the order-two quantum Wasserstein distance \cite{Magron2026}. Finite convergence has been established for generic multihomogeneous objectives on these domains \cite{HalasehMagronSkomra2025}. Combining the kernel approach with state-discrimination duality, we construct SOS certificates for channel coding by applying kernels directly to bounded measurable positive dual slack fields, which need not be polynomial. We then correct the matrix constraints and the trace bound simultaneously using the common mean and degree-two harmonic structure of the channel outputs.

% squared-polynomial kernels yield quadratic convergence of SOS hierarchies on the sphere, including for matrix-valued polynomials \cite{FangFawzi2021}. The kernel approach also gives quadratic convergence on products of spheres, with an application to the order-two quantum Wasserstein distance \cite{Magron2026}. Finite convergence has been established for generic multihomogeneous objectives on these domains \cite{HalasehMagronSkomra2025}. Our analysis combines the spherical kernel estimates in \cite{FangFawzi2021} with state-discrimination duality. Applying the kernels directly to bounded measurable positive dual slack fields and correcting the constraints through the common mean and degree-two harmonic structure of the channel outputs yields an error bound proportional to the advantage over random guessing.

\paragraph{Organization.}
\Cref{sec:background} presents the notation, coding problem, SOS cones, and spherical kernels. \Cref{sec:hierarchy} constructs the hierarchy and its finite SDP representation. \Cref{sec:convergence-analysis} proves the quadratic convergence bounds and their consequences for trace-norm contraction. \Cref{sec:numerics} compares SDP sizes and reports numerical experiments. The appendices provide the Gram, quotient-ideal, regularity, and kernel estimates used in the analysis.

%%%%%%%%%%%%%
\section{Background}
\label{sec:background}

This section fixes notation and records the two general tools used below: sum-of-squares certificates on products of real spheres and positive squared-kernel operators.  Throughout, all Hilbert spaces are finite dimensional and all measures are Borel.

\subsection{Basic notation}
\label{sec:prelim-channels}

Throughout, $A$ and $B$ are finite-dimensional complex Hilbert spaces,
with $d_A=\dim A$ and $d_B=\dim B$. We identify $A$ with $\C^{d_A}$.
The complex vector spaces of linear operators on $A$ and linear maps
from $A$ to $B$ are denoted by $\cL(A)$ and $\cL(A,B)$, respectively.
We use $\Herm(A)$ for the real vector space of Hermitian operators on
$A$, and write $X\succeq Y$ when $X-Y$ is positive semidefinite.
The set of density operators on $A$ is $ \cD(A)=\{\rho\in\Herm(A):\rho\succeq0,\ \Tr\rho=1\}.$
% \[
%  \cD(A)=\{\rho\in\Herm(A):\rho\succeq0,\ \Tr\rho=1\}.
% \]
We denote the identity operator on $A$ by $I_A$ and the identity map on
$\cL(A)$ by $\Id_A$. For every positive integer $k$, we write $[k]:=\{1,\ldots,k\}$.

For a Hermitian operator $X$, we write $\lmin(X)$ and $\lmax(X)$ for its
smallest and largest eigenvalues, respectively. For an arbitrary
$X\in\cL(A)$, its absolute value is $\abs{X}=\sqrt{X^\dagger X}$.
We denote its operator norm by $\norm{X}_\infty$, which equals its
largest singular value, and its trace norm by
$\norm{X}_1=\Tr\abs{X}$.

A \emph{quantum channel} $\Phi:\cL(A)\to\cL(B)$ is a completely positive,
trace-preserving linear map. Every quantum channel can be expressed
in terms of Kraus operators $K_1,\ldots,K_r\in\cL(A,B)$ as
\begin{equation}
 \Phi(X)=\sum_{a=1}^r K_aXK_a^\dagger,
 \qquad
 \sum_{a=1}^r K_a^\dagger K_a=I_A.
 \label{eq:kraus}
\end{equation}
Fix an orthonormal basis $\{\ket{j}\}_{j=1}^{d_A}$ of $A$.
The \emph{normalized Choi matrix} of $\Phi$ is defined by
\begin{equation}
 J(\Phi)=(\Id_A\otimes\Phi)(\proj{\Omega}),
 \qquad
 \ket{\Omega}=\frac{1}{\sqrt{d_A}}
              \sum_{j=1}^{d_A}\ket{j}\ket{j};
 \label{eq:choi}
\end{equation}
More generally, Eq.~\eqref{eq:choi} defines the Choi matrix of any linear map $\Phi:\cL(A)\to\cL(B)$. Complete positivity of $\Phi$ guarantees that
$J(\Phi)\succeq0$; trace preservation is not required for this property.

\subsection{Problem statement}
\label{sec:problem-results}

Consider a classical channel described by transition probabilities
$W(y\mid x)$ between finite input and output alphabets. Given an integer
$k\ge2$, the optimal channel coding problem asks for encoding and decoding
rules that maximize the average probability of correctly recovering one of
$k$ equiprobable messages after a single channel use. As shown in \cite{BarmanFawzi2018}, achieving a constant
approximation ratio strictly larger than $1-e^{-1}$ is NP-hard when $k$ is part of the input.

We study the corresponding problem for a quantum channel
$\Phi:\cL(A)\to\cL(B)$. To transmit message $i\in[k]$, the sender prepares
a state $\rho_i\in\cD(A)$ and sends it through $\Phi$. The receiver applies
a positive operator-valued measure (POVM) $(M_1,\ldots,M_k)$ to the output
and interprets outcome $i$ as the decoded message. The optimal average
success probability is therefore
\begin{equation}
 \Psucc(\Phi,k)
 =\max_{\substack{\rho_1,\ldots,\rho_k\in\cD(A)\\
                   M_1,\ldots,M_k\succeq0,\ \sum_iM_i=I_B}}
 \frac1k\sum_{i=1}^k\Tr\!\left[M_i\Phi(\rho_i)\right].
 \label{eq:Psucc}
\end{equation}
Classical channels can be represented within this model without changing
their optimal success probabilities. Consequently, the classical
hardness result also establishes NP-hardness of computing
$\Psucc(\Phi,k)$ when $k$ is part of the input.

For every channel $\Phi$ and every $k\ge2$,
\begin{equation}
 \frac1k\le\Psucc(\Phi,k)\le1.
 \label{eq:operational-bounds}
\end{equation}
The lower bound is achieved by random guessing: choosing $M_i=I_B/k$
gives $\Tr[M_i\Phi(\rho_i)]=1/k$ for every input state $\rho_i$.
For the upper bound, each POVM element satisfies
$0\preceq M_i\preceq I_B$, so
$\Tr[M_i\Phi(\rho_i)]\le\Tr\Phi(\rho_i)=1$.
Averaging these inequalities gives $\Psucc(\Phi,k)\le1$.

For two messages, the optimal success probability has a useful
characterization in terms of the \emph{trace-norm contraction coefficient}
of $\Phi$, defined by
\begin{equation}
 \etatr(\Phi)
 =\sup_{\substack{\rho,\sigma\in\cD(A)\\ \rho\ne\sigma}}
 \frac{\norm{\Phi(\rho)-\Phi(\sigma)}_1}
      {\norm{\rho-\sigma}_1}
 \in[0,1].
 \label{eq:etatr-def}
\end{equation}
When $d_A=1$, there is only one input state, and we set
$\etatr(\Phi)=0$. For $d_A\ge2$, the supremum in
\eqref{eq:etatr-def} is attained at a pair of orthogonal pure states
~\cite{DelsolEtAl2025}. Combining this characterization with the
Holevo--Helstrom theorem gives
\begin{equation}
 \Psucc(\Phi,2)=\frac{1+\etatr(\Phi)}2;
 \label{eq:P-eta}
\end{equation}
see~\cite{DelsolEtAl2025}. This identity also holds when $d_A=1$,
since the output is independent of the message and both sides equal
$\tfrac12$.

Although the classical problem can be solved efficiently for $k=2$,
the quantum problem remains NP-hard even in this case.
Indeed, \cite{DelsolEtAl2025} established NP-hardness of constant-factor
approximation of $\etatr(\Phi)$ and of deciding whether
$\etatr(\Phi)=1$. Through \eqref{eq:P-eta}, these results imply
NP-hardness of computing the optimal success probability for transmitting
a single bit, including deciding whether perfect transmission is possible.

% A previous approach to obtaining computable upper bounds is the
% semidefinite programming hierarchy of \citet{DelsolEtAl2025}, based on
% the framework for constrained bilinear optimization developed by
% \citet{BertaEtAl2021}. For binary messages, its level-$m$ value
% $\mathrm{SDP}_m(\Phi,k)$ satisfies
% \[
%  0\le
%  \mathrm{SDP}_m(\Phi,k)-\Psucc(\Phi,k)
%  \le\frac{\poly(d)}{\sqrt m},
%  \qquad d=\max\{d_A,d_B\}.
% \]
% Thus the hierarchy provides convergent upper bounds with an explicit
% a priori error guarantee. We compare its level and SDP-size scalings
% with those of our SOS hierarchy in \cref{sec:complexity}.

Using the framework for constrained bilinear optimization of
\cite{BertaEtAl2021}, \cite{DelsolEtAl2025} constructed a convergent
hierarchy of semidefinite programming upper bounds on
$\Psucc(\Phi,k)$ for an arbitrary number of messages $k\ge2$.
Its level-$m$ optimal value, denoted by $\mathrm{SDP}_m(\Phi,k)$,
satisfies
\[
 0\le
 \mathrm{SDP}_m(\Phi,k)-\Psucc(\Phi,k)
 \le\frac{\poly(d)}{\sqrt m},
 \qquad d=\max\{d_A,d_B\}.
\]
This gives an explicit a priori guarantee on the approximation error
at each level. We compare the level and SDP-size scalings of this
construction with those of our SOS hierarchy in \cref{sec:complexity}.

\subsection{Sum-of-squares certificates on products of spheres}
\label{sec:prelim-sos}

We recall the definitions of scalar and Hermitian matrix sum-of-squares certificates on products of spheres and fix the notation used below. These certificates are expressed modulo the sphere equations and will be used to replace pointwise positivity conditions by semidefinite constraints in the hierarchy. For further background on sum-of-squares certificates and polynomial optimization, including matrix-valued polynomials on spheres, we refer
to \cite{Laurent2009,FangFawzi2021}.

Fix integers $D\ge2$ and $k\ge1$. Let
$\sphere^{D-1}=\{v\in\R^D:\norm{v}_2=1\}$ and define
\begin{equation}
 \cX=(\sphere^{D-1})^k,
 \qquad
 x=(x_1,\ldots,x_k),\quad x_i\in\R^D.
 \label{eq:product-sphere}
\end{equation}
We regard $x$ as a vector of $k$ blocks, each containing $D$ coordinates.
For $i\in[k]$, the projection onto the $i$th block is
\begin{equation}
 \pi_i:\R^{kD}\to\R^D,
 \qquad
 \pi_i(x_1,\ldots,x_k)=x_i.
 \label{eq:block-projection}
\end{equation}

All polynomial variables represent real coordinates, although their
coefficients may be real or complex. We write
\[
 \R[x]=\R[x_1,\ldots,x_k],
 \qquad
 \C[x]=\C[x_1,\ldots,x_k],
\]
where each $x_i$ denotes a block of $D$ indeterminates.
On $\C[x]$, the involution $*$ conjugates coefficients and fixes the
indeterminates:
\[
 \left(\sum_\alpha c_\alpha x^\alpha\right)^*
 =\sum_\alpha\overline{c_\alpha}\,x^\alpha.
\]
For a matrix polynomial $F\in\C[x]^{s\times t}$, we define $F^\dagger$
by transposing $F$ and applying $*$ to each entry. In particular,
\[
 F^\dagger(x)=F(x)^\dagger
 \qquad\text{for every }x\in\R^{kD}.
\]

A matrix polynomial $F\in\C[x]^{s\times s}$ is \emph{Hermitian} if
$F^\dagger=F$. Equivalently, $F(x)\in\Herm(\C^s)$ for every
$x\in\R^{kD}$. In the coefficient expansion
\[
 F(x)=\sum_\alpha F_\alpha x^\alpha,
\]
this condition is equivalent to
$F_\alpha\in\Herm(\C^s)$ for every $\alpha$.
Consequently, for any fixed $C\in\Herm(\C^s)$, the scalar polynomial
$\Tr(CF)$ belongs to $\R[x]$.

For a scalar polynomial $p$, let $\deg_{x_i}p$ denote its total degree
in the coordinates of the block $x_i$. For a matrix polynomial $F$,
let $\deg_{x_i}F$ be the maximum of these degrees over its entries.
The corresponding \emph{block degrees} are
\begin{equation}
 \bdeg(p)=\max_{i\in[k]}\deg_{x_i}p,
 \qquad
 \bdeg(F)=\max_{i\in[k]}\deg_{x_i}F.
 \label{eq:bdeg}
\end{equation}

The equations defining $\cX$ generate the \emph{sphere ideal}
\begin{equation}
 \cI_{\cX}
 =\big\langle g_1,\ldots,g_k\big\rangle_{\R[x]},
 \qquad
 g_j(x_j)=\norm{x_j}_2^2-1.
 \label{eq:sphere-ideal}
\end{equation}
Its complexification is
\begin{equation}
 \cI_{\cX}^{\C}
 :=\big\langle g_1,\ldots,g_k\big\rangle_{\C[x]}
 =\left\{\sum_{j=1}^k g_jr_j:\ r_j\in\C[x]\right\}.
 \label{eq:complex-sphere-ideal}
\end{equation}
Since the generators have real coefficients,
\[
 \cI_{\cX}^{\C}=\cI_{\cX}+\mathrm{i}\cI_{\cX},
\]
and $\cI_{\cX}^{\C}$ is invariant under $*$, so it is a $*$-ideal.

For real scalar polynomials, congruence modulo $\cI_{\cX}$ means that
their difference belongs to $\cI_{\cX}$. For matrix polynomials, we
interpret ideal membership entrywise. Specifically, for $s\ge1$, set
\[
 M_s(\cI_{\cX}^{\C})
 :=
 \left\{
 H\in\C[x]^{s\times s}:
 H_{ab}\in\cI_{\cX}^{\C}\text{ for all }a,b
 \right\}.
\]
For $F,G\in\C[x]^{s\times s}$, we then write
\begin{equation}
 F\equiv G\pmod{\cI_{\cX}^{\C}}
 \quad\Longleftrightarrow\quad
 F-G\in M_s(\cI_{\cX}^{\C}).
 \label{eq:matrix-congruence}
\end{equation}
Equivalently, there exist matrix polynomials
$S_j\in\C[x]^{s\times s}$ such that
\[
 F-G=\sum_{j=1}^k g_jS_j.
\]

For an ambient scalar or matrix polynomial $F$ on $\R^{kD}$, we denote
its restriction to $\cX$ by $\left.F\right|_{\cX}$.
By \cref{lem:real-radical}, two real scalar polynomials agree on $\cX$
if and only if they are congruent modulo $\cI_{\cX}$.
Likewise, two complex matrix polynomials agree on $\cX$ if and only if
they are congruent modulo $\cI_{\cX}^{\C}$ in the entrywise sense of
\eqref{eq:matrix-congruence}.

\begin{definition}
\label{def:SOS-cones}
Fix an integer $L\ge1$. The scalar sum-of-squares cone $\Sigma_L(\cX)$
consists of the polynomials $p\in\R[x]$ for which there exist finitely
many $f_r\in\R[x]$, with $\bdeg(f_r)\le L$, satisfying
\begin{equation}
 p\equiv\sum_r f_r^2\pmod{\cI_{\cX}}.
 \label{eq:scalar-SOS}
\end{equation}

For $s\ge1$, the Hermitian matrix sum-of-squares cone
$\Sigma_L^{\Herm,s}(\cX)$ consists of the Hermitian matrix polynomials
$F\in\C[x]^{s\times s}$ for which there exist finitely many matrix
polynomials $R_r\in\C[x]^{s\times m_r}$, with $m_r\ge1$ and
$\bdeg(R_r)\le L$, satisfying
\begin{equation}
 F\equiv\sum_rR_rR_r^\dagger
 \pmod{\cI_{\cX}^{\C}}.
 \label{eq:matrix-SOS}
\end{equation}
\end{definition}

Every generator $g_j$ vanishes on $\cX$. Thus
$p\in\Sigma_L(\cX)$ implies $p(x)\ge0$ for every $x\in\cX$, while
$F\in\Sigma_L^{\Herm,s}(\cX)$ implies $F(x)\succeq0$ for every
$x\in\cX$.

\begin{remark}
\label{rem:real-complex-scalar-sos}
Using real coefficients for scalar SOS certificates is consistent with
allowing complex coefficients in matrix SOS factors. Indeed, if
$h=a+\mathrm{i}b\in\C[x]$, with $a,b\in\R[x]$, then $ hh^*=a^2+b^2.$
% \[
%  hh^*=a^2+b^2.
% \]
Hence scalar sums of Hermitian squares over $\C[x]$ coincide with
ordinary sums of squares over $\R[x]$, with the same degree bound.
This equivalence also holds modulo the sphere ideals, since $ \cI_{\cX}^{\C}\cap\R[x]=\cI_{\cX}.$

\end{remark}

\subsection{Positive squared kernels on the sphere}
\label{sec:prelim-kernels}

We recall the spherical harmonic decomposition and the squared-kernel
operators used in the convergence analysis. Their action on degree-two
harmonics will determine the approximation error of the hierarchy.
For further background on spherical harmonics and polynomial kernels
on the sphere, we refer to \cite{DunklXu2014,FangFawzi2021}.

Let $\mu_D$ be the unique rotation-invariant probability measure on
$\sphere^{D-1}$. For a fixed $x\in\sphere^{D-1}$, let $\nu_D$ denote
the distribution of $\ip{x}{y}$ when $y\sim\mu_D$. By rotation
invariance, this distribution is independent of the choice of $x$.

For $j\ge0$, let $\Pharm_j(\R^D)$ be the real vector space of
homogeneous harmonic polynomials of degree $j$ on $\R^D$. Their
restrictions to the sphere form the space of degree-$j$ spherical
harmonics,
\begin{equation}
 \cH_j^D
 :=\left\{
 \left.p\right|_{\sphere^{D-1}}:
 p\in\Pharm_j(\R^D)
 \right\}.
 \label{eq:ambient-spherical-harmonics}
\end{equation}
The restriction map in \eqref{eq:ambient-spherical-harmonics} is
injective. Indeed, if a homogeneous polynomial $p$ of degree $j$
vanishes on the sphere, then
\[
 p(x)=\norm{x}_2^j
 p\!\left(\frac{x}{\norm{x}_2}\right)=0
 \qquad\text{for every }x\ne0,
\]
and hence $p$ vanishes identically. Thus every element of $\cH_j^D$
has a unique ambient homogeneous harmonic representative.

The spaces $\cH_j^D$ are mutually orthogonal and give the orthogonal
decomposition of the real Hilbert space
\[
 L^2(\mu_D)=\widehat{\bigoplus}_{j\ge0}\cH_j^D;
\]
see \cite{DunklXu2014}. Rotation-invariant integral operators preserve
these harmonic subspaces. For a real vector space $\mathcal V$, we
write $\Pharm_j(\R^D)\otimes\mathcal V$ for the space of ambient
$\mathcal V$-valued homogeneous harmonic polynomial fields and
$\cH_j^D\otimes\mathcal V$ for their restrictions to the sphere.
Below, we will use this notation with $\mathcal V=\Herm(\C^s)$.

For each $j\ge0$, define the normalized zonal polynomial
\[
 g_{j,D}(t)=
 \begin{cases}
 \displaystyle
 \frac{C_j^{(D-2)/2}(t)}{C_j^{(D-2)/2}(1)},
 &D\ge3,\\[2ex]
 T_j(t),
 &D=2,
 \end{cases}
\]
where $C_j^\alpha$ is the Gegenbauer polynomial with parameter
$\alpha$ and $T_j$ is the Chebyshev polynomial of the first kind.
In both cases, $g_{j,D}(1)=1$.

Let $L\ge1$ and let $q$ be a real univariate polynomial of degree
at most $L$. Its associated \emph{squared-kernel operator} is
\begin{equation}
 (K_qf)(x)
 =\int_{\sphere^{D-1}}
 q(\ip{x}{y})^2f(y)\,d\mu_D(y),
 \qquad x\in\sphere^{D-1},
 \label{eq:single-kernel}
\end{equation}
defined for bounded measurable functions
$f:\sphere^{D-1}\to\C$. Since the kernel is real valued,
$K_q$ maps real-valued functions to real-valued functions.

We use the same notation for the entrywise extension of $K_q$ to
bounded measurable matrix-valued fields. In particular, for
$H:\sphere^{D-1}\to\Herm(\C^s)$,
\[
 (K_qH)(x)
 =\int_{\sphere^{D-1}}
 q(\ip{x}{y})^2H(y)\,d\mu_D(y)
 \in\Herm(\C^s).
\]
Thus the matrix-valued action is obtained by applying the scalar
operator to each matrix entry.

We normalize $q$ by requiring
\begin{equation}
 \int_{-1}^1q(t)^2\,d\nu_D(t)=1.
 \label{eq:kernel-normalization}
\end{equation}
This condition is equivalent to $K_q1=1$ on $\sphere^{D-1}$.
The Funk--Hecke formula \cite{DunklXu2014,FangFawzi2021} shows that
$K_q$ acts on $\cH_j^D$ by multiplication by
\begin{equation}
 \lambda_j(q)
 =\int_{-1}^1q(t)^2g_{j,D}(t)\,d\nu_D(t).
 \label{eq:funk-hecke-eigenvalue}
\end{equation}
The same formula holds on the complexification of $\cH_j^D$ and,
entrywise, on Hermitian matrix-valued harmonics. In particular,
\[
 K_qH=\lambda_j(q)H
 \qquad
 \text{for }H\in\cH_j^D\otimes\Herm(\C^s).
\]
Since $g_{0,D}=1$, we have
$\lambda_0(q)=\int q^2\,d\nu_D$, so the normalization
\eqref{eq:kernel-normalization} is precisely the condition
$\lambda_0(q)=1$.

The degree-two eigenvalue will be central to the convergence proof.
The corresponding zonal polynomial is
\begin{equation}
 g_{2,D}(t)=\frac{Dt^2-1}{D-1}.
 \label{eq:g2}
\end{equation}
At factor degree $L$, define the largest attainable degree-two
eigenvalue and the associated \emph{distortion constant} by
\begin{equation}
 \begin{aligned}
 \lambda_{D,L}
 &:=
 \max\left\{
 \lambda_2(q):
 q\in\R[t],\ \deg q\le L,\
 q\text{ satisfies }\eqref{eq:kernel-normalization}
 \right\},\\
 \rho_{D,L}
 &:=\lambda_{D,L}^{-1}-1.
 \end{aligned}
 \label{eq:lambda-rho-definition}
\end{equation}
To see that the maximum is attained, expand $q$ in an orthonormal
polynomial basis of $L^2(\nu_D)$ through degree $L$. These basis
polynomials are normalized Gegenbauer polynomials for $D\ge3$ and
normalized Chebyshev polynomials for $D=2$. If
$e\in\R^{L+1}$ is the coefficient vector, then
\[
 \int q^2\,d\nu_D=\norm{e}_2^2,
 \qquad
 \lambda_2(q)=e^\top T_{D,L}e,
\]
where $T_{D,L}$ is the real symmetric matrix described in
\cref{app:kernel-constant}. Maximization over $\norm{e}_2=1$
therefore gives
\[
 \lambda_{D,L}=\lmax(T_{D,L}).
\]

\begin{proposition}
\label{prop:kernel-constant}
For all $D\ge2$ and $L\ge1$, we have
\begin{equation}
\frac{2L}{D+2L}\le\lambda_{D,L}\le1,
\qquad
0\le\rho_{D,L}\le\frac{D}{2L}.
 \label{eq:lambda-basic}
\end{equation}
Moreover, there exist absolute constant $c_1>0$
such that
\begin{equation}
 \rho_{D,L}\le c_1\left(\frac{D}{L}\right)^2
 \qquad\text{for all }L\ge 1.
 \label{eq:FF-rate}
\end{equation}
\end{proposition}

% The lower bound on $\lambda_{D,L}$ follows by choosing
% $q(t)=\sqrt D\,t$ and using the moments
% \[
%  \int_{-1}^1t^2\,d\nu_D(t)=\frac1D,
%  \qquad
%  \int_{-1}^1t^4\,d\nu_D(t)=\frac{3}{D(D+2)}.
% \]
The lower bound in \eqref{eq:lambda-basic} follows from the normalized monomial choice $q(t)\propto t^L$, for which $\lambda_2(q)=2L/(D+2L)$; see \cref{app:kernel-constant}.
% The estimate \eqref{eq:FF-rate} is the degree-two specialization
% of the kernel bound in \cite{FangFawzi2021}, from which the constants
% $c_0,c_1$ are inherited. The details are given in
% \cref{app:kernel-constant}. The structural arguments below require
% only $\lambda_{D,L}>0$, which \eqref{eq:lambda-basic} guarantees
% at every level $L\ge1$. The constants $c_0,c_1$ are needed only
% for the quantitative convergence bounds.
For sufficiently large $L$ relative to $D$, the estimate \eqref{eq:FF-rate} follows from the degree-two kernel bound in \cite{FangFawzi2021} for $D\ge3$ and a direct Chebyshev argument for $D=2$. The normalized monomial choice $q(t)\propto t^L$ extends the estimate to every $L\ge1$ after enlarging the absolute constant $c_1$ if necessary. Details are given in \cref{app:kernel-constant}. The structural arguments below require only $\lambda_{D,L}>0$, which \eqref{eq:lambda-basic} guarantees at every level $L\ge1$; the constant $c_1$ enters only the quantitative convergence bounds.

Finally, we extend the construction to the product space $\cX$
defined in \eqref{eq:product-sphere}. Set
\begin{equation}
 \begin{aligned}
 \cK_q&:=K_q^{\otimes k},\\
 (\cK_qf)(x)
 &=
 \int_{\cX}
 \left(\prod_{j=1}^kq(\ip{x_j}{y_j})^2\right)
 f(y)\,d\mu_D^{\otimes k}(y),
 \qquad x\in\cX.
 \end{aligned}
 \label{eq:product-kernel}
\end{equation}
As with $K_q$, this formula defines $\cK_q$ on bounded measurable
scalar functions, and its action on matrix-valued fields is
understood entrywise. In particular, $\cK_q$ preserves Hermitian
matrix-valued fields. Linearity of the integral also gives
\begin{equation}
 \Tr\!\left[(\cK_qH)(x)\right]
 =\bigl(\cK_q(\Tr H)\bigr)(x).
 \label{eq:kernel-trace-compatibility}
\end{equation}

Under \eqref{eq:kernel-normalization}, $\cK_q$ fixes constant scalar
and matrix fields. Its action on a harmonic field depending on a
single block follows by applying \eqref{eq:funk-hecke-eigenvalue}
in that block. For example, if
$H\in\cH_2^D\otimes\Herm(\C^s)$, then
\[
 \bigl(\cK_q(H\circ\pi_i)\bigr)(x)
 =\lambda_2(q)H(x_i),
 \qquad x\in\cX.
\]
The same identity holds for scalar degree-two harmonics and their
complexifications. Indeed, by Fubini's theorem, integration over
the active block gives the factor $\lambda_2(q)$, while every
other block contributes the factor $\lambda_0(q)=1$.

\section{The sum-of-squares hierarchy}
\label{sec:hierarchy}

We construct an SOS hierarchy for $\Psucc(\Phi,k)$ by first restricting the optimization to pure input states and expressing them in real coordinates. We then use state-discrimination duality to obtain an exact infinite-dimensional formulation and replace its pointwise positivity constraints by finite-degree SOS certificates.

\subsection{Reduction to pure input states}
\label{sec:pure-reduction}

Fix a POVM $(M_1,\ldots,M_k)$. For each message $i$, we have $\Tr[M_i\Phi(\rho_i)]=\Tr[\Phi^*(M_i)\rho_i]$, where $\Phi^*:\cL(B)\to\cL(A)$ is the adjoint of $\Phi$ with respect to the Hilbert--Schmidt inner product. Since $\Phi^*$ is positive, $\Phi^*(M_i)$ is Hermitian positive semidefinite. It follows that $\max_{\rho_i\in\cD(A)}\Tr[\Phi^*(M_i)\rho_i]=\lmax(\Phi^*(M_i))$, with the maximum attained by a rank-one projector onto a corresponding unit eigenvector. Equivalently, a linear functional on the compact convex set $\cD(A)$ attains its maximum at an extreme point, and these extreme points are precisely the pure states.

For a fixed POVM, the input states can be optimized independently. Each $\rho_i$ can therefore be replaced by a pure state without decreasing the objective. Since pure states are feasible in the original problem, restricting the encoder to pure states leaves the joint optimum unchanged. Consequently,
\begin{equation}
 \Psucc(\Phi,k)=\max_{\substack{u_1,\ldots,u_k\in\C^{d_A}\\ \norm{u_i}_2=1}}\max_{\substack{M_1,\ldots,M_k\succeq0\\ \sum_iM_i=I_B}}\frac1k\sum_{i=1}^k\Tr\!\left[M_i\Phi(\proj{u_i})\right].
 \label{eq:pure-state-Psucc}
\end{equation}
The optimization over the encoder is thus parameterized by $k$ complex unit vectors, which we next express in real polynomial coordinates.

\subsection{Parametrization by real spheres}
\label{sec:sphere-parametrization}

Set $D=2d_A$. For $v=(a,b)\in\R^{d_A}\times\R^{d_A}\simeq\R^D$, define the $\R$-linear map $u:\R^D\to\C^{d_A}$ by
\begin{equation}
 u(v)=a+\mathrm{i}b,
 \qquad
 \norm{u(v)}_2=\norm{v}_2.
 \label{eq:realification-map}
\end{equation}
Thus the complex unit vectors representing pure input states are parameterized by $\sphere^{D-1}$. The global-phase freedom in this representation is immaterial, since the channel output depends only on the projector $u(v)u(v)^\dagger$.

Define the ambient Hermitian matrix polynomial $\widetilde\sigma:\R^D\to\Herm(B)$ by $\widetilde\sigma(v):=\Phi(u(v)u(v)^\dagger)$. For each message $i\in[k]$, its dependence on the corresponding input block is described by
\begin{equation}
 \widetilde\sigma_i:=\widetilde\sigma\circ\pi_i,
 \qquad
 \sigma_i:=\left.\widetilde\sigma_i\right|_{\cX},
 \label{eq:sigma-definition}
\end{equation}
where $\pi_i$ is the block projection defined in \eqref{eq:block-projection} and $\cX=(\sphere^{D-1})^k$. In particular, $\sigma_i:\cX\to\cD(B)$ and $\sigma_i(x)=\Phi(\proj{u(x_i)})$ for every $x\in\cX$.

Each entry of $\widetilde\sigma_i$ is a complex-valued homogeneous quadratic polynomial in the real coordinates of $x_i$. Its restriction $\sigma_i$ has trace one on $\cX$, but the ambient polynomial need not be normalized away from $\cX$. Indeed, trace preservation gives the polynomial identity
\begin{equation}
 \Tr\widetilde\sigma_i(x)=\norm{x_i}_2^2,
 \qquad x\in\R^{kD}.
 \label{eq:ambient-sigma-trace}
\end{equation}

\subsection{State-discrimination duality and an infinite optimization problem}
\label{sec:exact-dual}

For fixed states $\omega_1,\ldots,\omega_k\in\cD(B)$ with uniform prior, the optimal success probability for minimum-error discrimination is
\begin{equation}
s(\omega_1,\ldots,\omega_k)=\max\left\{\frac1k\sum_{i=1}^k\Tr(M_i\omega_i):\ M_i\succeq0,\ \sum_{i=1}^kM_i=I_B\right\}.
\label{eq:discrimination-primal}
\end{equation}
The dual semidefinite program gives the equivalent expression \cite{Watrous2018}
\begin{equation}
s(\omega_1,\ldots,\omega_k)=\min\left\{\Tr Y:\ Y\in\Herm(B),\ Y\succeq\frac{\omega_i}{k}\ \text{for all }i\in[k]\right\}.
\label{eq:discrimination-dual}
\end{equation}
Indeed, the primal problem is strictly feasible with $M_i=I_B/k$, and the dual problem is strictly feasible with $Y=tI_B$ for sufficiently large $t$. Strong duality therefore holds, and both optima are attained. Every optimal dual solution satisfies $Y\succeq\omega_i/k\succeq0$ and $\Tr Y=s(\omega_1,\ldots,\omega_k)\le1$. Since the largest eigenvalue of a positive semidefinite operator is at most its trace, it follows that $0\preceq Y\preceq I_B$.

For the output states associated with the sphere parameters, define the value function
\begin{equation}
s:\cX\to\left[\tfrac1k,1\right],\qquad s(x):=s\bigl(\sigma_1(x),\ldots,\sigma_k(x)\bigr).
\label{eq:sx}
\end{equation}
Combining \eqref{eq:pure-state-Psucc} with \eqref{eq:discrimination-dual} yields
\begin{equation}
\Psucc(\Phi,k)=\max_{x\in\cX}s(x),\qquad s(x)=\min_{\substack{Y\in\Herm(B)}} \{ \Tr Y: Y\succeq\sigma_i(x)/k,\ i\in[k] \} \, .
\label{eq:parametric-dual}
\end{equation}
This eliminates the optimization over the decoding POVM exactly. The remaining nonconvexity lies in the dependence of the output states $\sigma_i(x)$ on the sphere parameters.

The pointwise dual problems can be combined into a single infinite-dimensional optimization problem:
\begin{equation}
\begin{aligned}
\Psucc(\Phi,k)=\inf_{\gamma\in\R,\ Y(\cdot)}\quad &\gamma\\
\text{subject to}\quad &Y(x)\succeq\frac{\sigma_i(x)}k, &&x\in\cX,\ i\in[k],\\
&\gamma-\Tr Y(x)\ge0, &&x\in\cX.
\end{aligned}
\label{eq:infinite-dual-envelope}
\end{equation}
Here $Y(\cdot)$ ranges over all fields $\cX\to\Herm(B)$, with no regularity assumption at this stage. To verify the equality, observe that every feasible pair satisfies $\gamma\ge\Tr Y(x)\ge s(x)$ for all $x\in\cX$, and hence $\gamma\ge\Psucc(\Phi,k)$. Conversely, choose an optimal dual solution $Y(x)$ for each $x\in\cX$. Then $\Tr Y(x)=s(x)\le\Psucc(\Phi,k)$, so this field together with $\gamma=\Psucc(\Phi,k)$ is feasible and attains the claimed value.

\subsection{The finite SOS hierarchy and its SDP representation}
\label{sec:finite-hierarchy}

Restricting the matrix field in \eqref{eq:infinite-dual-envelope} to a polynomial of bounded block degree and replacing its pointwise positivity constraints by SOS certificates gives the following hierarchy.

\begin{definition}
\label{def:hierarchy}
For an integer $L\ge1$, define
\begin{equation}
\begin{aligned}
U_L(\Phi,k):=\inf_{\gamma\in\R,\ Y}\quad &\gamma\\
\text{subject to}\quad &Y-\frac1k\widetilde\sigma_i\in\Sigma_L^{\Herm,d_B}(\cX), &&i\in[k],\\
&\gamma-\Tr Y\in\Sigma_L(\cX),\\
&Y^\dagger=Y,\quad \bdeg(Y)\le2L,
\end{aligned}
\label{eq:UL-definition}
\end{equation}
where $Y\in\C[x]^{d_B\times d_B}$ and $\widetilde\sigma_i$ is the ambient quadratic matrix polynomial defined in \eqref{eq:sigma-definition}. The SOS constraints are understood as congruences between ambient polynomials modulo the corresponding sphere ideals, as in \cref{sec:prelim-sos}.
\end{definition}

The SDP representation follows from the Gram representation and realification properties established in \cref{lem:realification}. Let $z_L(x)\in\R^{N_{D,L,k}}$ be the column vector of all monomials of block degree at most $L$, where
\begin{equation}
N_{D,L,k}=\binom{D+L}{L}^{\!k}.
\label{eq:monomial-count}
\end{equation}
A Hermitian matrix polynomial $F$ of size $s$ belongs to $\Sigma_L^{\Herm,s}(\cX)$ if and only if there exists a Hermitian positive semidefinite matrix $Q\in\Herm(\C^{N_{D,L,k}s})$ such that
\begin{equation}
F\equiv(z_L^\top\otimes I_s)\,Q\,(z_L\otimes I_s)\pmod{\cI_{\cX}^{\C}}.
\label{eq:Gram-representation}
\end{equation}
For the scalar cone, the corresponding representation is $p\equiv z_L^\top Q_0z_L\pmod{\cI_{\cX}}$, with $Q_0$ real symmetric and positive semidefinite. Moreover, writing a Hermitian matrix as $Q=Q_{\mathrm R}+\mathrm{i}Q_{\mathrm I}$, where $Q_{\mathrm R}^\top=Q_{\mathrm R}$ and $Q_{\mathrm I}^\top=-Q_{\mathrm I}$, gives
\[
Q\succeq0\quad\Longleftrightarrow\quad
\begin{pmatrix}Q_{\mathrm R}&-Q_{\mathrm I}\\ Q_{\mathrm I}&Q_{\mathrm R}\end{pmatrix}\succeq0.
\]
Thus each complex Hermitian semidefinite constraint can be expressed as a real symmetric semidefinite constraint of twice the order. The degree control needed to express the polynomial congruences by finitely many affine equations is established in \cref{lem:quotient-degree}.

\begin{proposition}
\label{prop:finite-SDP}
For every $L\ge1$, problem \eqref{eq:UL-definition} admits a finite-dimensional SDP representation. The scalar Gram matrix is real symmetric of order $N_{D,L,k}$, and each of the $k$ matrix Gram matrices is complex Hermitian of order $d_BN_{D,L,k}$. After realification, each complex Hermitian Gram block has real symmetric order $2d_BN_{D,L,k}$. With $D=2d_A$, one has $N_{D,L,k}=(d_A+L)^{O(kL)}$, which is polynomial in $d_A$ for fixed $k$ and $L$.
\end{proposition}

\begin{proof}
The decision variables consist of the coefficients of $Y$, the $k$ Hermitian positive semidefinite Gram matrices $Q_i$ for the matrix slacks, a real symmetric positive semidefinite Gram matrix $Q_0$ for the scalar slack, and polynomial multipliers for the sphere constraints. By \eqref{eq:Gram-representation}, the matrix constraints are equivalent to the existence of $Q_i\succeq0$ and $S_{ij}\in\C[x]^{d_B\times d_B}$ satisfying
\[
Y-\frac1k\widetilde\sigma_i=(z_L^\top\otimes I_{d_B})Q_i(z_L\otimes I_{d_B})+\sum_{j=1}^kg_jS_{ij},\qquad i\in[k].
\]
Similarly, the scalar constraint is equivalent to the existence of $Q_0\succeq0$ and $r_j\in\R[x]$ satisfying
\[
\gamma-\Tr Y=z_L^\top Q_0z_L+\sum_{j=1}^kg_jr_j.
\]
The matrix multipliers $S_{ij}$ encode membership in $M_{d_B}(\cI_{\cX}^{\C})$ and carry no positivity constraints. Since each remainder has block degree at most $2L$, \cref{lem:quotient-degree} allows the multiplier of $g_j$ to be chosen with degree at most $2L-2$ in block $x_j$ and at most $2L$ in every other block, entrywise for matrix multipliers. All polynomial variables therefore have finitely many coefficients. Matching coefficients gives affine equations; for complex identities, their real and imaginary parts are matched separately. Together with the Gram constraints and the realification described above, these equations give a finite-dimensional SDP. The Gram matrix orders follow from \eqref{eq:monomial-count}.
\end{proof}

The SOS constraints ensure that every level gives an upper bound on the optimal success probability. These bounds are nonincreasing with the level and satisfy the same upper bound of one as the original problem.

\begin{proposition}
\label{prop:basic-properties}
For every $L\ge1$, we have
\begin{equation}
\Psucc(\Phi,k)\le U_L(\Phi,k)\le1,\qquad U_{L+1}(\Phi,k)\le U_L(\Phi,k).
\label{eq:basic-hierarchy-bounds}
\end{equation}
\end{proposition}

\begin{proof}
Let $(\gamma,Y)$ be feasible for \eqref{eq:UL-definition}. Evaluating the SOS constraints on $\cX$ gives $Y(x)\succeq\sigma_i(x)/k$ for every $i\in[k]$ and $\gamma\ge\Tr Y(x)$. By \eqref{eq:parametric-dual}, it follows that $\gamma\ge s(x)$ for every $x\in\cX$, and hence $\gamma\ge\Psucc(\Phi,k)$. Taking the infimum proves the lower bound. Monotonicity follows because both the SOS cones and the admissible polynomial spaces are nested as $L$ increases.

To prove $U_L(\Phi,k)\le1$, set $Y_{\mathrm{av}}:=\frac1k\sum_{j=1}^k\widetilde\sigma_j$, which is Hermitian and has block degree at most two. The Kraus representation \eqref{eq:kraus} gives
\begin{equation}
\widetilde\sigma_j(x)=\sum_{a=1}^r\bigl(K_au(x_j)\bigr)\bigl(K_au(x_j)\bigr)^\dagger.
\label{eq:sigma-HSOS}
\end{equation}
Each vector $K_au(x_j)$ has entries that are complex linear forms in the real coordinates of $x_j$, so $\widetilde\sigma_j\in\Sigma_1^{\Herm,d_B}(\cX)$. Consequently, for every $i\in[k]$,
\[
Y_{\mathrm{av}}-\frac1k\widetilde\sigma_i=\frac1k\sum_{j\ne i}\widetilde\sigma_j\in\Sigma_1^{\Herm,d_B}(\cX)\subseteq\Sigma_L^{\Herm,d_B}(\cX).
\]
Furthermore, the ambient trace identity \eqref{eq:ambient-sigma-trace} implies
\[
1-\Tr Y_{\mathrm{av}}=-\frac1k\sum_{j=1}^k\bigl(\norm{x_j}_2^2-1\bigr)\in\cI_{\cX}.
\]
Thus $1-\Tr Y_{\mathrm{av}}$ is congruent to zero modulo $\cI_{\cX}$ and belongs to $\Sigma_L(\cX)$. The pair $(1,Y_{\mathrm{av}})$ is therefore feasible at every level $L\ge1$, proving the upper bound.
\end{proof}

\section{Quadratic convergence}
\label{sec:convergence-analysis}

% The bounds in \cref{prop:basic-properties} show that $U_L(\Phi,k)$ is a nonincreasing upper bound on $\Psucc(\Phi,k)$. To prove convergence, it suffices to construct, for each level $L$, a feasible pair $(\gamma_L,Y_L)$ for \eqref{eq:UL-definition} whose objective approaches $\Psucc(\Phi,k)$. Indeed, feasibility gives $\Psucc(\Phi,k)\le U_L(\Phi,k)\le\gamma_L$, so an estimate on $\gamma_L-\Psucc(\Phi,k)$ directly controls the hierarchy's error. This section constructs such a pair using positive squared kernels and establishes a quadratic convergence rate.
The bounds in \cref{prop:basic-properties} show that $U_L(\Phi,k)$ is a nonincreasing upper bound on $\Psucc(\Phi,k)$. To prove convergence, it suffices to construct, for each level $L$, a feasible pair $(\hat{\gamma}_L,\hat{Y}_L)$ for \eqref{eq:UL-definition} whose objective approaches $\Psucc(\Phi,k)$. Indeed, feasibility gives $\Psucc(\Phi,k)\le U_L(\Phi,k)\le\hat{\gamma}_L$, so an estimate on $\hat{\gamma}_L-\Psucc(\Phi,k)$ directly controls the hierarchy's error. This section constructs such a pair using positive squared kernels and establishes a quadratic convergence rate.

Throughout this section, $\cK_q$ denotes the product kernel defined in \cref{sec:prelim-kernels}, including its entrywise action on Hermitian matrix-valued fields. The argument of $\cK_q$ determines whether its action is scalar or matrix-valued.

\medskip
\noindent\textbf{Difficulties in proving convergence.} The first difficulty is that pointwise optimal dual solutions of \eqref{eq:infinite-dual-envelope} need not form a polynomial field. Before applying the kernel, a Borel measurable choice of these solutions is required.

\begin{lemma}
\label{lem:measurable-selector}
The function $s:\cX\to\R$ defined in \eqref{eq:sx} is continuous, and there exists a Borel measurable map $Y_*:\cX\to\Herm(B)$ such that, for every $x\in\cX$,
\begin{equation}
Y_*(x)\succeq\frac{\sigma_i(x)}k\quad(i\in[k]),\qquad \Tr Y_*(x)=s(x),\qquad 0\preceq Y_*(x)\preceq I_B.
\label{eq:Ystar-properties}
\end{equation}
\end{lemma}

The proof uses Berge's maximum theorem and the measurable maximum theorem \cite{AliprantisBorder2006} and is given in \cref{app:selection}. That appendix also provides an alternative argument using the Kuratowski--Ryll-Nardzewski selection theorem \cite{KuratowskiRyllNardzewski1965}.

The binary case illustrates the difficulty with polynomial dependence and motivates the use of positive kernels. Suppose $k=2$ and, for a fixed $x\in\cX$, suppress the dependence on $x$. Write $\Delta=\sigma_1-\sigma_2$ and let $\Delta=\Delta_+-\Delta_-$ be its Jordan decomposition, where $\Delta_\pm\succeq0$ have orthogonal supports. A pointwise optimal solution of \eqref{eq:discrimination-dual} is
\begin{equation}
Y_*=\frac14\left(\sigma_1+\sigma_2+\abs{\Delta}\right),\qquad Y_*-\frac{\sigma_1}{2}=\frac{\Delta_-}{2},\qquad Y_*-\frac{\sigma_2}{2}=\frac{\Delta_+}{2}.
\label{eq:binary-dual-formula}
\end{equation}
The identities $\Delta_\pm=(\abs{\Delta}\pm\Delta)/2$ verify feasibility, while $\Tr Y_*=\frac12+\frac14\norm{\Delta}_1$ equals the Holevo--Helstrom value \cite{Helstrom1969} and therefore proves optimality.

As a function of the sphere parameters, this solution involves $\abs{\Delta(x)}$, which need not be polynomial and can be nonsmooth when an eigenvalue of $\Delta(x)$ passes through zero. The inverse-kernel argument in \cite{FangFawzi2021} starts from a polynomial input and therefore does not apply directly to this field. A uniform polynomial approximation of $Y_*$ would introduce an additive error controlled by the modulus of continuity of the fields obtained from the Jordan decomposition; such an estimate alone would not preserve the desired multiplicative form. The construction in \cref{prop:channel-certificate} instead applies the kernel directly to the positive slack fields. In the binary case, these fields are obtained from the pointwise Jordan decomposition in \eqref{eq:binary-dual-formula}, and their smoothed versions admit exact polynomial SOS representations.

The second difficulty is that applying the kernel also changes the nonconstant components of the constraint fields $\sigma_i$. The correction relies on their common structure: all output fields have the same mean, and their deviations from this mean lie in a single degree-two harmonic eigenspace. Consequently, the kernel scales every deviation by the same scalar, allowing its effect on the constraints to be corrected exactly.

%%%%%%%%%%%%%%%%%%%%%%%%%%

\subsection{Kernel certificates and convergence}
\label{sec:kernel-certificate}
\label{sec:convergence}

We first establish the common mean and harmonic structure of the output fields and show that positive squared kernels produce finite-degree SOS certificates from positive measurable fields. These properties allow us to construct a feasible pair $(\hat{\gamma}_L,\hat{Y}_L)$ for \eqref{eq:UL-definition} and bound its objective value, yielding the quadratic convergence rate.

\medskip
\noindent\textbf{Common mean and harmonic structure.} The output fields $\sigma_i$ share the same spherical mean, and their deviations from this mean belong to a single degree-two harmonic eigenspace. Define
\begin{equation}
\tau:=\Phi\!\left(\frac{I_A}{d_A}\right)\in\cD(B).
\label{eq:tau-definition}
\end{equation}

\begin{lemma}
\label{lem:common-mean-harmonic}
Let $D=2d_A$. The ambient matrix polynomial
\begin{equation}
\widetilde h(v):=\widetilde\sigma(v)-\norm{v}_2^2\,\tau,\qquad v\in\R^D,
\label{eq:ambient-h-definition}
\end{equation}
belongs to $\Pharm_2(\R^D)\otimes\Herm(B)$. Let $h:=\left.\widetilde h\right|_{\sphere^{D-1}}\in\cH_2^D\otimes\Herm(B)$. Then
\begin{equation}
\sigma_i(x)=\tau+h(x_i),\qquad x\in\cX,\ i\in[k],
\label{eq:sigma-harmonic-decomposition}
\end{equation}
whereas the corresponding ambient identity is
\begin{equation}
\widetilde\sigma_i(x)=\norm{x_i}_2^2\,\tau+\widetilde h(x_i),\qquad x\in\R^{kD},\ i\in[k].
\label{eq:ambient-sigma-harmonic-decomposition}
\end{equation}
Consequently, for every real polynomial $q$ satisfying \eqref{eq:kernel-normalization}, with $\lambda:=\lambda_2(q)$, we have
\begin{equation}
(\cK_q\sigma_i)(x)=\tau+\lambda h(x_i),\qquad x\in\cX,\ i\in[k].
\label{eq:kernel-on-sigma}
\end{equation}
\end{lemma}

\begin{proof}
Write $v=(a,b)$ as in \eqref{eq:realification-map}. For $v\sim\mu_D$, the coordinate moments satisfy $\E[v_pv_r]=\delta_{pr}/D$. Thus $\E[a_ja_\ell]=\E[b_jb_\ell]=\delta_{j\ell}/D$ and $\E[a_jb_\ell]=0$. Since $(u(v)u(v)^\dagger)_{j\ell}=a_ja_\ell+b_jb_\ell+\mathrm{i}(b_ja_\ell-a_jb_\ell)$, it follows that
\[
\int_{\sphere^{D-1}}u(v)u(v)^\dagger\,d\mu_D(v)=\frac{2}{D}I_A=\frac{I_A}{d_A}.
\]
Applying $\Phi$ gives $\int_{\sphere^{D-1}}\widetilde\sigma(v)\,d\mu_D(v)=\tau$.

Every real or imaginary part of an entry of $\widetilde\sigma$ is a quadratic form $v^\top Sv$ with $S\in\R^{D\times D}$ symmetric. Its spherical mean is $\Tr(S)/D$, so subtracting $\norm{v}_2^2$ times this mean gives a quadratic form $v^\top S'v$ with $\Tr(S')=0$. Its Laplacian is $2\Tr(S')=0$. Applying this observation entrywise shows that $\widetilde h\in\Pharm_2(\R^D)\otimes\Herm(B)$. Restricting \eqref{eq:ambient-h-definition} to the sphere gives \eqref{eq:sigma-harmonic-decomposition}, and composing with $\pi_i$ gives \eqref{eq:ambient-sigma-harmonic-decomposition}.

Finally, $\cK_q$ fixes the constant field $\tau$ and multiplies the degree-two harmonic in the active block by $\lambda_2(q)$. Each inactive block contributes the factor $\lambda_0(q)=1$, which proves \eqref{eq:kernel-on-sigma}.
\end{proof}

\medskip
\noindent\textbf{Kernel--Gram representations.} The next lemma describes how the kernel produces polynomial and SOS certificates from measurable fields. For bounded inputs, the integral formula \eqref{eq:product-kernel} also makes sense for $x\in\R^{kD}$ and defines an ambient polynomial. When discussing polynomial degrees, identities, or SOS membership, $\cK_qH$ denotes this ambient polynomial; its restriction to $\cX$ agrees with the operator defined in \cref{sec:prelim-kernels}.

\begin{lemma}
\label{lem:positive-kernel}
Let $q$ be a real polynomial with $\deg q\le L$.
\begin{enumerate}[label=(\roman*)]
\item If $H:\cX\to\Herm(\C^s)$ is bounded and Borel measurable, then $\cK_qH$ is a Hermitian matrix polynomial with $\bdeg(\cK_qH)\le2L$.
\item If, in addition, $H(x)\succeq0$ for all $x\in\cX$, then
\begin{equation}
\cK_qH\in\Sigma_L^{\Herm,s}(\cX).
\label{eq:kernel-matrix-SOS}
\end{equation}
\item If $f:\cX\to\R$ is bounded and Borel measurable, then $\cK_qf$ is a real polynomial with $\bdeg(\cK_qf)\le2L$. If moreover $f\ge0$ on $\cX$, then $\cK_qf\in\Sigma_L(\cX)$.
\end{enumerate}
\end{lemma}

\begin{proof}
For $y\in\cX$, define $a_y(x):=\prod_{j=1}^kq(\ip{x_j}{y_j})$. This is a real polynomial in $x$ of block degree at most $L$, with coefficients polynomial in $y$. Hence $a_y(x)=c(y)^\top z_L(x)$ for a bounded continuous map $c:\cX\to\R^{N_{D,L,k}}$. The kernel in \eqref{eq:product-kernel} equals $a_y(x)^2$, so
\begin{equation}
\begin{aligned}
(\cK_qH)(x)&=(z_L(x)^\top\otimes I_s)\,Q_H\,(z_L(x)\otimes I_s),\\
Q_H&:=\int_{\cX}c(y)c(y)^\top\otimes H(y)\,d\mu_D^{\otimes k}(y).
\end{aligned}
\label{eq:kernel-finite-Gram}
\end{equation}
The integral converges absolutely because $c$ and $H$ are bounded. Moreover, $Q_H\in\Herm(\C^{N_{D,L,k}s})$ because $c$ is real valued and $H$ is Hermitian valued. Thus \eqref{eq:kernel-finite-Gram} is an exact ambient polynomial identity that exhibits $\cK_qH$ as a Hermitian matrix polynomial of block degree at most $2L$, proving (i).

For (ii), take $\xi\in\C^{N_{D,L,k}s}$ and set $w_\xi(y):=(c(y)^\top\otimes I_s)\xi$. If $H(y)\succeq0$ for every $y\in\cX$, then
\[
\xi^\dagger Q_H\xi=\int_{\cX}w_\xi(y)^\dagger H(y)w_\xi(y)\,d\mu_D^{\otimes k}(y)\ge0.
\]
Consequently, $Q_H\succeq0$, and \eqref{eq:kernel-finite-Gram} is a Gram representation of the form \eqref{eq:Gram-representation}, with exact equality of ambient polynomials. This proves \eqref{eq:kernel-matrix-SOS}. The same calculation with a real scalar input gives a real symmetric Gram matrix for $\cK_qf$, proving (iii).
\end{proof}

Identity \eqref{eq:kernel-finite-Gram} is the standard kernel--Gram construction underlying squared-kernel arguments on the sphere \cite{FangFawzi2021,Magron2026}. The hypotheses in \cref{lem:positive-kernel} distinguish its two consequences: boundedness and measurability suffice for polynomiality, while pointwise positivity additionally gives an SOS certificate. Neither conclusion requires the input field to be polynomial or continuous, and the SOS conclusion requires no strict positivity margin. Since the integral itself is a finite positive semidefinite Gram matrix, no closure argument for an infinite SOS cone is needed. In particular, the lemma applies to the optimal dual field of~\cref{lem:measurable-selector}, which need not be polynomial.

\medskip
\noindent\textbf{A feasible SOS certificate.} Write $P:=\Psucc(\Phi,k)$. We apply $\cK_q$ to the optimal dual field $Y_*$, the positive matrix slacks $Y_*-\sigma_i/k$, and the nonnegative scalar slack $P-s$. Its scalar and matrix-valued actions are compatible with the trace by \eqref{eq:kernel-trace-compatibility}. By \eqref{eq:kernel-on-sigma}, smoothing preserves the common mean $\tau$ but multiplies each degree-two component by $\lambda_2(q)$. The following affine correction compensates for this change while preserving the SOS level.

\begin{proposition}
\label{prop:channel-certificate}
Let $q$ be a normalized real polynomial with $\deg q\le L$ and $\lambda:=\lambda_2(q)>0$. For this fixed choice of $q$, define
\begin{equation}
\hat{Y}_L:=\frac1\lambda\cK_qY_*-\frac{1-\lambda}{\lambda k}\tau,\qquad \hat{\gamma}_L:=\frac{P}{\lambda}-\frac{1-\lambda}{\lambda k}.
\label{eq:channel-certificate}
\end{equation}
Then $\hat{Y}_L$ is an ambient Hermitian matrix polynomial with $\bdeg(\hat{Y}_L)\le2L$, and
\begin{equation}
\hat{Y}_L-\frac1k\widetilde\sigma_i\in\Sigma_L^{\Herm,d_B}(\cX),\qquad i\in[k],
\label{eq:channel-matrix-certificate}
\end{equation}
together with
\begin{equation}
\hat{\gamma}_L-\Tr\hat{Y}_L\in\Sigma_L(\cX).
\label{eq:channel-scalar-certificate}
\end{equation}
Consequently, $(\hat{\gamma}_L,\hat{Y}_L)$ is feasible for \eqref{eq:UL-definition}. Writing $\rho:=\lambda^{-1}-1$, its objective value is
\begin{equation}
\hat{\gamma}_L=P+\rho\left(P-\frac1k\right).
\label{eq:channel-certificate-value}
\end{equation}
\end{proposition}

\begin{proof}
By \cref{lem:positive-kernel}(i), $\cK_qY_*$ is an ambient Hermitian matrix polynomial of block degree at most $2L$. The same therefore holds for $\hat{Y}_L$.

Fix $i\in[k]$. On $\cX$, equations \eqref{eq:sigma-harmonic-decomposition} and \eqref{eq:kernel-on-sigma} give
\[
\frac1\lambda\cK_q\!\left(\frac{\sigma_i}{k}\right)=\frac{\tau}{\lambda k}+\frac{h\circ\pi_i}{k}.
\]
Substituting this identity into \eqref{eq:channel-certificate} yields
\begin{equation}
\hat{Y}_L-\frac{\sigma_i}{k}=\frac1\lambda\cK_q\!\left(Y_*-\frac{\sigma_i}{k}\right)\qquad\text{on }\cX.
\label{eq:channel-slack-on-X}
\end{equation}
By \eqref{eq:Ystar-properties}, the field $Y_*-\sigma_i/k$ is bounded, Borel measurable, and pointwise positive semidefinite. Since $\lambda>0$, \cref{lem:positive-kernel}(ii) implies
\[
\frac1\lambda\cK_q\!\left(Y_*-\frac{\sigma_i}{k}\right)\in\Sigma_L^{\Herm,d_B}(\cX).
\]
The left-hand side of \eqref{eq:channel-slack-on-X} is the restriction of $\hat{Y}_L-\widetilde\sigma_i/k$, and the right-hand side is the restriction of the ambient kernel polynomial. Their difference vanishes on $\cX$, so \cref{lem:real-radical} gives
\begin{equation}
\hat{Y}_L-\frac1k\widetilde\sigma_i\equiv\frac1\lambda\cK_q\!\left(Y_*-\frac{\sigma_i}{k}\right)\pmod{\cI_{\cX}^{\C}}.
\label{eq:channel-slack-congruence}
\end{equation}
This proves \eqref{eq:channel-matrix-certificate}.

For the scalar constraint, trace compatibility gives $\Tr(\cK_qY_*)=\cK_q(\Tr Y_*)=\cK_qs$. Since $\Tr\tau=1$, equation \eqref{eq:channel-certificate} implies $\Tr\hat{Y}_L=\lambda^{-1}\cK_qs-(1-\lambda)/(\lambda k)$. Using $\cK_q1=1$ on $\cX$, we obtain
\begin{equation}
\hat{\gamma}_L-\Tr\hat{Y}_L=\frac1\lambda\cK_q(P-s)\qquad\text{on }\cX.
\label{eq:channel-objective-on-X}
\end{equation}
The field $P-s$ is bounded, Borel measurable, and nonnegative because $P=\max_{x\in\cX}s(x)$. Hence \cref{lem:positive-kernel}(iii) places the right-hand side in $\Sigma_L(\cX)$. Applying \cref{lem:real-radical} to \eqref{eq:channel-objective-on-X} proves \eqref{eq:channel-scalar-certificate}. Finally,
\[
\frac{P}{\lambda}-\frac{1-\lambda}{\lambda k}=P+\rho\left(P-\frac1k\right),
\]
which proves \eqref{eq:channel-certificate-value}.
\end{proof}

The scaling by $\lambda^{-1}$ restores the degree-two components of the channel constraints, while subtracting $(1-\lambda)\tau/(\lambda k)$ restores their common constant component. The construction therefore obtains finite Gram certificates directly from the positive slack fields $Y_*-\sigma_i/k$ and $P-s$, using the same kernel with entrywise action in the matrix-valued case. It requires no norm approximation of the generally nonsmooth selector $Y_*$.

\medskip
\noindent\textbf{Main convergence.} Choosing a normalized kernel polynomial that maximizes its degree-two eigenvalue gives the following bounds.

\begin{theorem}
\label{thm:main-convergence}
Let $\Phi:\cL(\C^{d_A})\to\cL(\C^{d_B})$ be a quantum channel and let $k\ge2$. For every $L\ge1$, we have
\begin{equation}
\Psucc(\Phi,k)\le U_L(\Phi,k)\le\min\!\left\{1,\ \Psucc(\Phi,k)+\rho_{2d_A,L}\left(\Psucc(\Phi,k)-\frac1k\right)\right\}.
\label{eq:main-convergence-bound}
\end{equation}
Consequently,
\begin{equation}
0\le U_L(\Phi,k)-\Psucc(\Phi,k)\le\frac{\rho_{2d_A,L}}{1+\rho_{2d_A,L}}\left(1-\frac1k\right).
\label{eq:uniform-main-gap}
\end{equation}
Moreover, with the absolute constant $c_1$ from \cref{prop:kernel-constant}, for every $L\ge 1$, we have
\begin{equation}
U_L(\Phi,k)-\Psucc(\Phi,k)\le4c_1\frac{d_A^2}{L^2}\left(\Psucc(\Phi,k)-\frac1k\right)\le4c_1\left(1-\frac1k\right)\frac{d_A^2}{L^2}.
\label{eq:explicit-relative-rate}
\end{equation}
\end{theorem}

\begin{proof}
Write $P=\Psucc(\Phi,k)$ and choose a normalized polynomial $q_L$ attaining $\lambda_{2d_A,L}$ in \eqref{eq:lambda-rho-definition}. By \cref{prop:kernel-constant}, $\lambda_{2d_A,L}>0$. Applying \cref{prop:channel-certificate} with $q=q_L$ gives a feasible pair $(\hat{\gamma}_L,\hat{Y}_L)$ satisfying
\[
U_L(\Phi,k)\le\hat{\gamma}_L=P+\rho_{2d_A,L}\left(P-\frac1k\right).
\]
Combining this inequality with $P\le U_L(\Phi,k)\le1$ from \cref{prop:basic-properties} proves \eqref{eq:main-convergence-bound}.

For \eqref{eq:uniform-main-gap}, set $\rho=\rho_{2d_A,L}$. The gap is at most $\min\{1-P,\rho(P-1/k)\}$. Over $P\in[1/k,1]$, the two affine functions intersect at $P_*=(1+\rho/k)/(1+\rho)$, where their common value is $\rho(1-1/k)/(1+\rho)$. This value bounds their minimum for every $P$ in the interval.

Finally, \cref{prop:kernel-constant} with $D=2d_A$ gives $\rho_{2d_A,L}\le c_1(2d_A/L)^2=4c_1d_A^2/L^2$ for every $L\ge 1$. Substitution into \eqref{eq:main-convergence-bound}, together with $P\le1$, yields \eqref{eq:explicit-relative-rate}.
\end{proof}

The error in \eqref{eq:main-convergence-bound} is relative to the advantage over random guessing, $\Psucc(\Phi,k)-1/k$. This dependence follows directly from the baseline $\Tr(\tau/k)=1/k$ in the affine correction. The error bound also has no dependence on the output dimension $d_B$: the kernel acts on the input parameters in $\sphere^{2d_A-1}$, and \cref{lem:positive-kernel} preserves positive semidefiniteness without a loss depending on matrix size. This is the same matrix-size independence as in the matrix-valued sphere result of \cite{FangFawzi2021}. 
%The output dimension enters the SDP size, as discussed in \cref{sec:complexity}.

\begin{corollary}
\label{cor:P-endpoints}
If $\Psucc(\Phi,k)\in\{1/k,1\}$, then $U_L(\Phi,k)=\Psucc(\Phi,k)$ for every $L\ge1$.
\end{corollary}

\begin{proof}
If $\Psucc(\Phi,k)=1/k$, the relative term in \eqref{eq:main-convergence-bound} vanishes. If $\Psucc(\Phi,k)=1$, its lower and upper bounds coincide.
\end{proof}

%%%%%%%%%%%%%%%%%%%%%%%%%%%%%%%%%%5

\subsection{Binary messages and trace-norm contraction}
\label{sec:binary}

For binary messages, the identity \eqref{eq:P-eta} converts the bounds on $\Psucc(\Phi,2)$ into multiplicative bounds on the trace-norm contraction coefficient. Define
\begin{equation}
E_L(\Phi):=2U_L(\Phi,2)-1.
\label{eq:EL-definition}
\end{equation}

\begin{theorem}
\label{thm:contraction-hierarchy}
For every $L\ge1$, we have
\begin{equation}
\etatr(\Phi)\le E_L(\Phi)\le\min\!\left\{1,\ (1+\rho_{2d_A,L})\etatr(\Phi)\right\}.
\label{eq:eta-multiplicative}
\end{equation}
Moreover, with the absolute constant $c_1$ from \cref{prop:kernel-constant}, for every $L\ge 1$, one has
\begin{equation}
E_L(\Phi)-\etatr(\Phi)\le4c_1\frac{d_A^2}{L^2}\etatr(\Phi).
\label{eq:eta-explicit-rate}
\end{equation}
If $\etatr(\Phi)\in\{0,1\}$, then $E_L(\Phi)=\etatr(\Phi)$ for every $L\ge1$.
\end{theorem}

\begin{proof}
By \eqref{eq:P-eta}, $\Psucc(\Phi,2)-1/2=\etatr(\Phi)/2$, including the case $d_A=1$ under the convention following \eqref{eq:etatr-def}. Hence $E_L(\Phi)-\etatr(\Phi)=2(U_L(\Phi,2)-\Psucc(\Phi,2))$. Applying \cref{thm:main-convergence} with $k=2$ gives \eqref{eq:eta-multiplicative} and \eqref{eq:eta-explicit-rate}, with $U_L(\Phi,2)\le1$ yielding $E_L(\Phi)\le1$. The endpoint statement follows from \cref{cor:P-endpoints}, since $\etatr(\Phi)=0$ and $\etatr(\Phi)=1$ correspond to $\Psucc(\Phi,2)=1/2$ and $\Psucc(\Phi,2)=1$, respectively.
\end{proof}

% When the upper bound of one is inactive, \eqref{eq:eta-multiplicative} gives a multiplicative approximation with factor $1+\rho_{2d_A,L}$. This is particularly useful in the strong-contraction regime, where an additive error independent of $\etatr(\Phi)$ may be large relative to the quantity being estimated. For example, when $A=B$, the trace distance after $n$ successive applications of $\Phi$ is bounded by $\etatr(\Phi)^n$ times the initial trace distance. The bound \eqref{eq:eta-multiplicative} controls the relative error uniformly over positive values of $\etatr(\Phi)$ and is exact when $\etatr(\Phi)=0$.
The bound \eqref{eq:eta-multiplicative} gives a multiplicative approximation with factor $1+\rho_{2d_A,L}$. This is particularly useful in the strong-contraction regime, where an additive error independent of $\etatr(\Phi)$ may be large relative to the quantity being estimated. For example, when $A=B$, the trace distance after $n$ successive applications of $\Phi$ is bounded by $\etatr(\Phi)^n$ times the initial trace distance. The bound \eqref{eq:eta-multiplicative} controls the relative error uniformly over positive values of $\etatr(\Phi)$ and is exact when $\etatr(\Phi)=0$.

\section{Numerical experiments}
\label{sec:numerics}

This section compares the SOS and extension hierarchies in terms of SDP size and first-level performance. We first examine the largest PSD block orders at a common guaranteed additive error. We then compare bounds for binary messages on an exactly solvable benchmark and random qubit-to-qutrit channels, including the extension hierarchy with PPT constraints.

\subsection{SDP size at a prescribed accuracy}
\label{sec:complexity}

We compare the largest PSD block orders at levels guaranteeing an additive relaxation error of at most $\varepsilon\in(0,1)$. Let $d=\max\{d_A,d_B\}$, and let $C_{\mathrm{ext}}(d)=\poly(d)>0$ be the prefactor in the extension error bound $C_{\mathrm{ext}}(d)/\sqrt m$ from \cite{DelsolEtAl2025}. Together with \cref{thm:main-convergence}, this gives the sufficient levels
\begin{equation}
 L_\varepsilon=\left\lceil\frac{C_{\mathrm{SOS}}d_A}{\sqrt{\varepsilon}}\right\rceil,
 \qquad
 m_\varepsilon=\left\lceil\frac{C_{\mathrm{ext}}(d)^2}{\varepsilon^2}\right\rceil,
 \qquad
 C_{\mathrm{SOS}}:=2\max\{1,\sqrt{c_1}\}.
\label{eq:L-of-eps}
\end{equation}
These levels follow from a priori bounds and may exceed those needed for particular channels.

Set $N_\varepsilon:=\binom{2d_A+L_\varepsilon}{L_\varepsilon}^{k}$. The largest SOS Gram block has order $d_BN_\varepsilon$ by \cref{prop:finite-SDP}, whereas the direct extension formulation has blocks of order $d_Ad_B^{m_\varepsilon}$ \cite{DelsolEtAl2025}. Applying the commutant reduction of \cite{CheeTaVu2025} to the classical message structure, while retaining all irreducible sectors, gives the reduced bound in \cref{tab:sizes}.

\begin{table}[h]
\centering
\small
\setlength{\tabcolsep}{6pt}
\renewcommand{\arraystretch}{1.3}
\begin{tabular}{|l|c|c|c|}
\hline
& SOS & $\mathrm{SDP}_m$ & $\mathrm{SDP}_m$ (symmetry)\\
\hline
Sufficient level
  & $L_\varepsilon$
  & $m_\varepsilon$
  & $m_\varepsilon$\\
\hline
Largest PSD block order
  & $d_B N_\varepsilon$
  & $d_A d_B^{\,m_\varepsilon}$
  & $\le d_A (m_\varepsilon+1)^{k d_B(d_B-1)/2}$\\
\hline
\end{tabular}
\caption{Largest PSD block orders at guaranteed additive error $\varepsilon$.}
\label{tab:sizes}
\end{table}

For fixed $d_A,d_B,k$, the SOS block order is $\Theta(\varepsilon^{-kd_A})$, and the symmetry-reduced extension block order is $O(\varepsilon^{-kd_B(d_B-1)})$. The direct extension block order is exponential in $\varepsilon^{-2}$ when $d_B\ge2$. At fixed accuracy and fixed $k$, the SOS block order grows linearly in $d_B$ when $d_A$ is fixed, whereas the symmetry-reduced extension block order grows polynomially in $d_A$ when $d_B$ is fixed. These comparisons concern the largest PSD block and do not by themselves determine runtime.

\subsection{First-level bounds for binary channel coding}
\label{sec:numerics-binary}

We compare the first SOS level with the first extension level from \cite{DelsolEtAl2025}, with and without positive partial transpose (PPT) constraints. All values are expressed on the contraction-coefficient scale:
\[
 \begin{aligned}
 E_1(\Phi)&=2U_1(\Phi,2)-1,\\
 E_1^{\mathrm{ext}}(\Phi)&=2\mathrm{SDP}_1(\Phi,2)-1,\\
 E_1^{\mathrm{ext,PPT}}(\Phi)
 &=2\mathrm{SDP}_1^{\mathrm{PPT}}(\Phi,2)-1.
 \end{aligned}
\]
We consider an exactly solvable benchmark and $40$ random qubit-to-qutrit channels.

\paragraph{An exactly solvable benchmark.}
Consider the channel from \cite{DelsolEtAl2025} with Kraus operators
\[
 K_{ij}=\frac{\ket{i}\!\bra{j}}{\sqrt{d-1}},
 \qquad i\ne j,
\]
whose action is
\[
 \Phi_d(X)
 =\frac{\Tr(X)I_d-\operatorname{diag}(X)}{d-1}.
\]
Trace-norm contractivity of diagonal pinching gives
$\etatr(\Phi_d)\le1/(d-1)$, with equality attained by two distinct computational-basis states.

For $d=3$, take $Y=I_3/4$ and $\gamma=3/4$. These satisfy
\[
 Y-\frac12\widetilde\sigma_i
 \equiv R_iR_i^\dagger
 \pmod{\cI_{\cX}^{\C}},
 \qquad
 R_i=\frac12\operatorname{diag}(u(x_i)),
 \qquad
 \gamma-\Tr Y=0.
\]
Since $\bdeg(R_i)=1$, this is a feasible first-level certificate and gives
\[
 \etatr(\Phi_3)\le E_1(\Phi_3)\le\frac12.
\]
Both extension formulations, with and without PPT constraints, give the same bound $\etatr(\Phi_3)\le1/2$ \cite{DelsolEtAl2025}.

\paragraph{Random channels and reference values.}
We sample $20$ channels of Choi rank $6$ and $20$ of Choi rank $3$ by taking a QR factorization of a complex Ginibre matrix in $\C^{3r\times2}$ and extracting Kraus operators from the resulting isometry. The samples use NumPy's PCG64 generator with $\texttt{SeedSequence}([r,s])$, where $r\in\{3,6\}$ and $s=0,\ldots,19$. The same channels are used in all three formulations.

For a qubit input, let $S_1,S_2,S_3$ be the Pauli matrices and set $H_j=\Phi(S_j)$. Then
\begin{equation}
 \etatr(\Phi)
 =\max_{n\in\sphere^2}
 \frac12\norm{\sum_{j=1}^3n_jH_j}_1.
 \label{eq:numerics-bloch-reference}
\end{equation}
We compute a reference value $\eta_{\mathrm{ref}}$ by evaluating this objective on $8192$ spherical grid points and refining $32$ selected directions by local optimization. This procedure gives an achievable lower estimate but does not certify global optimality.

The SOS implementation uses a phase-averaged representation at $L=1$. Averaging over independent phase rotations of the input vectors preserves feasibility, the degree bounds, and the objective, so this representation has the same optimum. The extension formulations follow \cite{DelsolEtAl2025}.

All $120$ random-channel solves return \texttt{optimal\_inaccurate}. We repair the SOS Gram matrices and coefficient residuals and adjust the objective accordingly. Writing $\gamma_{\mathrm{rep}}$ for the adjusted objective, we report
\[
 \hat{E}_1=2\gamma_{\mathrm{rep}}-1.
\]
The largest reported primal--dual gap across the three methods is $5.5\times10^{-8}$. All computations and feasibility checks use floating-point arithmetic, so the repaired values are not certified by interval arithmetic.

\paragraph{Results.}
All reported gaps are measured relative to $\eta_{\mathrm{ref}}$. As summarized in \cref{tab:numerics-random}, the computed SOS gaps satisfy
\[
 0\le\hat{E}_1-\eta_{\mathrm{ref}}
 \le4.10\times10^{-7}
\]
across all $40$ channels. The mean extension gaps are $0.101499$ without PPT and $0.101377$ with PPT. The repaired SOS value is smaller than both computed extension values in every instance, with a minimum observed separation greater than $0.0164$.

\begin{table}[H]
\centering
\small
\setlength{\tabcolsep}{6pt}
\renewcommand{\arraystretch}{1.3}
\begin{tabular}{|c|c|c|c|c|}
\hline
Choi rank & Samples
 & \shortstack{Maximum gap\\$\hat{E}_1-\eta_{\mathrm{ref}}$}
 & \shortstack{Mean gap\\$E_1^{\mathrm{ext}}-\eta_{\mathrm{ref}}$}
 & \shortstack{Mean gap\\$E_1^{\mathrm{ext,PPT}}-\eta_{\mathrm{ref}}$}\\
\hline
$6$ & $20$ & $1.55\times10^{-7}$ & $0.089232$ & $0.088999$\\
$3$ & $20$ & $4.10\times10^{-7}$ & $0.113766$ & $0.113756$\\
\hline
All & $40$ & $4.10\times10^{-7}$ & $0.101499$ & $0.101377$\\
\hline
\end{tabular}
\caption{First-level gaps for random qubit-to-qutrit channels.}
\label{tab:numerics-random}
\end{table}

The individual results in \cref{fig:numerics-qubit-qutrit} show numerical agreement between the repaired SOS values and the reference values throughout the sample. The first extension level retains a visible gap, with little improvement from PPT constraints. These observations support numerical tightness of the first SOS level on the sampled channels but do not establish its exactness for all qubit-to-qutrit channels.

\begin{figure}[H]
\centering
\includegraphics[scale = 0.45]{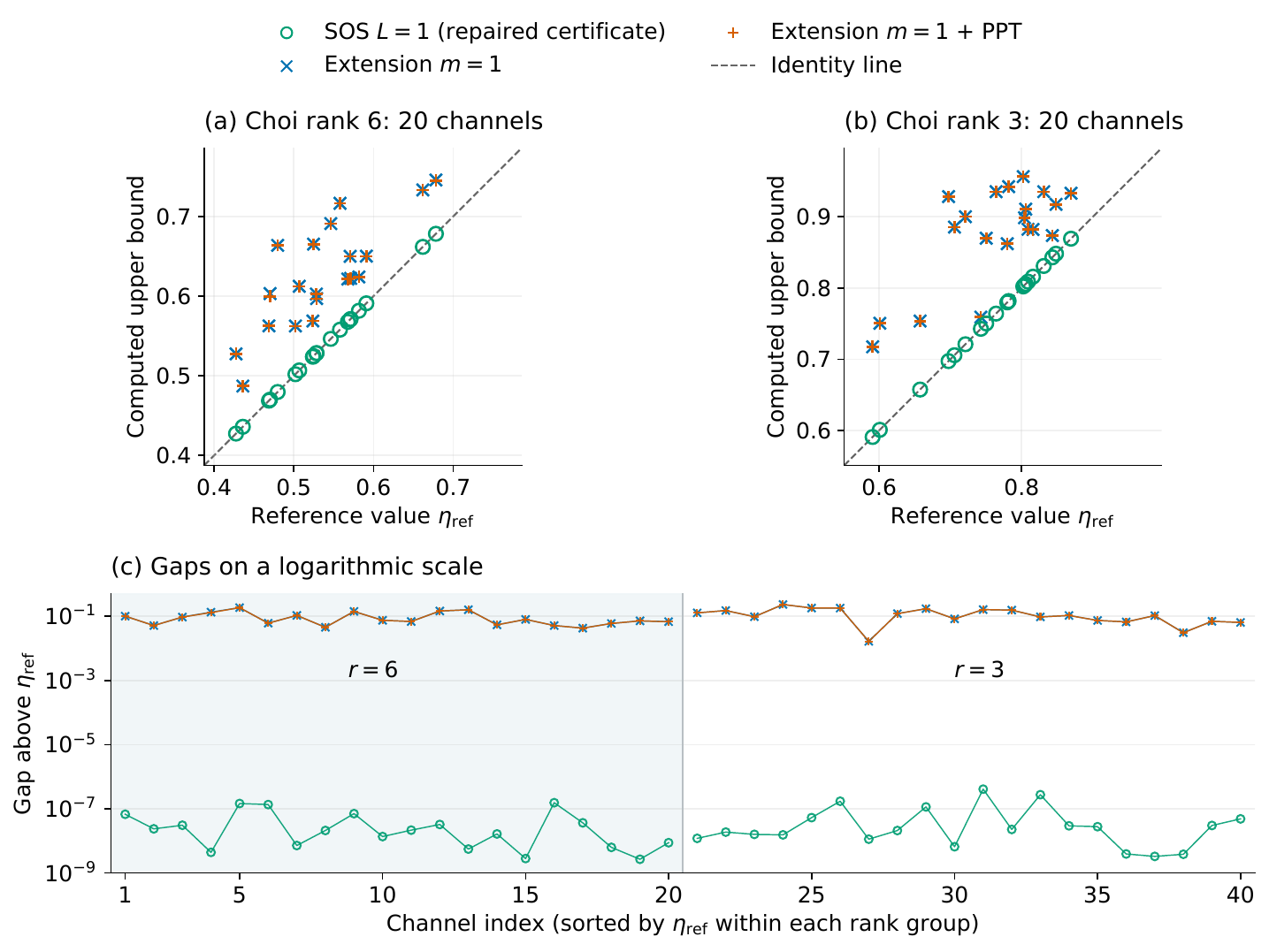}
\caption{First-level bounds for $40$ random qubit-to-qutrit channels. Panels (a) and (b) compare the computed values with $\eta_{\mathrm{ref}}$ for Choi ranks $6$ and $3$. Panel (c) shows the gaps on a logarithmic scale, with channels sorted by $\eta_{\mathrm{ref}}$ within each rank group. SOS values include numerical certificate repair.}
\label{fig:numerics-qubit-qutrit}
\end{figure}

\section*{Declaration of AI Usage}
The authors led the research direction and the development of the sum-of-squares hierarchy and its convergence analysis. AI models (ChatGPT 5.5 Pro and ChatGPT-5.6 Sol Pro) assisted with exploring mathematical arguments, drafting intermediate proof sketches, and organizing and drafting parts of the exposition. All AI-generated material was carefully checked and substantially revised by the authors, who take full responsibility for the manuscript and the validity of its results.

\section*{Acknowledgments} We would like to thank Prof.~Omar Fawzi for helpful discussions and valuable suggestions.

%\newpage

\bibliographystyle{alpha}
\bibliography{references}

%%%%%%%%%%%%%%%%%%%%%%%%%%%%%%%%%%%%%%%%%%%%%%%%%%%%%%%%%%%%%%%%%%%%%%%

\newpage
\appendix

\section{Hermitian matrix SOS cones and realification}
\label{app:HSOS}

This appendix establishes the equivalence between the matrix SOS definition in \cref{def:SOS-cones} and the Gram representation \eqref{eq:Gram-representation}. It also gives the realification used in the SDP representation of \cref{prop:finite-SDP}.

\begin{lemma}
\label{lem:realification}
Let $z=(z_1,\ldots,z_N)^\top\in\R[x]^N$ and let $s\ge1$.
\begin{enumerate}[label=(\roman*)]
\item For a Hermitian matrix polynomial $F\in\C[x]^{s\times s}$, the following statements are equivalent.
\begin{enumerate}[label=(\alph*)]
\item There exist finitely many matrix polynomials $R_r\in\C[x]^{s\times m_r}$, with $m_r\ge1$ and entries in $\operatorname{span}_{\C}\{z_1,\ldots,z_N\}$, such that
\begin{equation}
F=\sum_rR_rR_r^\dagger.
\label{eq:finite-HSOS-z}
\end{equation}
\item There exists $Q\in\Herm(\C^{Ns})$ with $Q\succeq0$ such that
\begin{equation}
F=(z^\top\otimes I_s)\,Q\,(z\otimes I_s).
\label{eq:exact-Gram-z}
\end{equation}
\end{enumerate}
In particular, taking $z=z_L$, a Hermitian matrix polynomial $F$ belongs to $\Sigma_L^{\Herm,s}(\cX)$ if and only if there exists $Q\in\Herm(\C^{N_{D,L,k}s})$, with $Q\succeq0$, satisfying \eqref{eq:Gram-representation}.

\item For $Q\in\Herm(\C^{Ns})$, write $Q=Q_{\mathrm R}+\mathrm{i}Q_{\mathrm I}$, where $Q_{\mathrm R},Q_{\mathrm I}\in\R^{Ns\times Ns}$ satisfy $Q_{\mathrm R}^\top=Q_{\mathrm R}$ and $Q_{\mathrm I}^\top=-Q_{\mathrm I}$. Then
\begin{equation}
Q\succeq0\quad\Longleftrightarrow\quad
\begin{pmatrix}Q_{\mathrm R}&-Q_{\mathrm I}\\ Q_{\mathrm I}&Q_{\mathrm R}\end{pmatrix}\succeq0
\quad\text{over }\R.
\label{eq:realification-PSD}
\end{equation}
\end{enumerate}
\end{lemma}

\begin{proof}
(i) Suppose that (b) holds. Factor $Q=VV^\dagger$ with $V\in\C^{Ns\times m}$ and define $R:=(z^\top\otimes I_s)V$. Every entry of $R$ belongs to $\operatorname{span}_{\C}\{z_1,\ldots,z_N\}$, and
\[
RR^\dagger=(z^\top\otimes I_s)VV^\dagger(z\otimes I_s)=F.
\]
Thus (a) holds with a single rectangular factor.

Conversely, suppose that (a) holds. For each $r$, the assumption on the entries of $R_r$ gives a coefficient matrix $V_r\in\C^{Ns\times m_r}$ such that $R_r=(z^\top\otimes I_s)V_r$. Setting $Q:=\sum_rV_rV_r^\dagger$ gives $Q\in\Herm(\C^{Ns})$, $Q\succeq0$, and
\[
(z^\top\otimes I_s)Q(z\otimes I_s)=\sum_rR_rR_r^\dagger=F.
\]
This proves the equivalence.

Now take $z=z_L$. If $F\in\Sigma_L^{\Herm,s}(\cX)$, then \cref{def:SOS-cones} gives factors $R_r$ with $\bdeg(R_r)\le L$ such that $F\equiv\sum_rR_rR_r^\dagger\pmod{\cI_{\cX}^{\C}}$. Every entry of each factor lies in the complex span of the entries of $z_L$. Applying the equivalence above to $\sum_rR_rR_r^\dagger$ therefore gives a positive semidefinite Gram matrix satisfying \eqref{eq:Gram-representation}. Conversely, factoring a positive semidefinite Gram matrix in \eqref{eq:Gram-representation} produces matrix SOS factors of block degree at most $L$, proving $F\in\Sigma_L^{\Herm,s}(\cX)$.

(ii) Define the real symmetric matrix
\[
\mathfrak R(Q):=\begin{pmatrix}Q_{\mathrm R}&-Q_{\mathrm I}\\ Q_{\mathrm I}&Q_{\mathrm R}\end{pmatrix}.
\]
For every $a,b\in\R^{Ns}$,
\begin{equation}
\begin{pmatrix}a\\ b\end{pmatrix}^{\!\top}\mathfrak R(Q)\begin{pmatrix}a\\ b\end{pmatrix}
=(a+\mathrm{i}b)^\dagger Q(a+\mathrm{i}b).
\label{eq:realification-quadratic-form}
\end{equation}
Every vector in $\C^{Ns}$ has a unique representation $a+\mathrm{i}b$. Hence the right-hand side is nonnegative for every complex vector if and only if the left-hand side is nonnegative for every real vector, proving \eqref{eq:realification-PSD}.
\end{proof}

The existence quantifier in part~(i) is essential. The linear Gram map $Q\mapsto(z^\top\otimes I_s)Q(z\otimes I_s)$ need not be injective, so a matrix polynomial may admit both positive semidefinite and indefinite Gram representations. Membership in the matrix SOS cone therefore does not imply that an arbitrarily prescribed Gram matrix is positive semidefinite.

\section{Degree control in the product-sphere quotient}
\label{app:quotient-degree}

The following results justify the finite coefficient constraints in \cref{sec:finite-hierarchy} and the passage from pointwise polynomial identities on $\cX$ to congruences modulo the sphere ideals.

\begin{lemma}
\label{lem:quotient-degree}
Let $L\ge1$ and $g_j(x_j)=\norm{x_j}_2^2-1$ for $j\in[k]$. If $p\in\cI_{\cX}\subset\R[x]$ satisfies $\deg_{x_i}p\le2L$ for every $i\in[k]$, then $p=\sum_{j=1}^kg_jr_j$ for polynomials $r_j\in\R[x]$ satisfying
\begin{equation}
\deg_{x_j}r_j\le2L-2,\qquad \deg_{x_i}r_j\le2L\quad(i\ne j).
\label{eq:ideal-multiplier-degree}
\end{equation}
The same statement holds over $\C[x]$. In particular, if $H\in M_s(\cI_{\cX}^{\C})$ satisfies $\deg_{x_i}H\le2L$ for every $i\in[k]$, then there exist $S_j\in\C[x]^{s\times s}$ such that
\[
H=\sum_{j=1}^kg_jS_j,\qquad \deg_{x_j}S_j\le2L-2,\qquad \deg_{x_i}S_j\le2L\quad(i\ne j),
\]
with the degree bounds understood entrywise.
\end{lemma}

\begin{proof}
Choose a lexicographic monomial order in which $x_{j,D}$ is the greatest variable within block $x_j$, for every $j\in[k]$. Then $\operatorname{LM}(g_j)=x_{j,D}^2$. These leading monomials are pairwise relatively prime, so Buchberger's product criterion implies that $G=\{g_1,\ldots,g_k\}$ is a Gr\"obner basis of $\cI_{\cX}$ \cite{CoxLittleOShea2015}.

The case $p=0$ follows by taking all multipliers to be zero. Otherwise, set $B_i:=\deg_{x_i}p\le2L$ and divide $p$ by $G$. Since $p\in\cI_{\cX}$, the remainder is zero. We show that every monomial appearing in the current dividend has degree at most $B_i$ in each block $x_i$ throughout the division.

This holds initially. Suppose that a term $t$ is reduced using $g_j$. The quotient term $m:=t/x_{j,D}^2$ satisfies $\deg_{x_j}m\le B_j-2$ and $\deg_{x_i}m\le B_i$ for $i\ne j$. After cancelling $t$, the newly introduced terms are scalar multiples of $m$ and $mx_{j,\ell}^2$ for $\ell<D$. Their degrees are at most $B_j$ in block $x_j$ and at most $B_i$ in every other block. The claim follows by induction.

Every term added to the quotient $r_j$ is obtained in this way. Thus $\deg_{x_j}r_j\le B_j-2\le2L-2$ and $\deg_{x_i}r_j\le B_i\le2L$ for $i\ne j$. If $B_j<2$, no reduction uses $g_j$ and $r_j=0$. Since the remainder is zero, these quotients satisfy $p=\sum_{j=1}^kg_jr_j$.

The division argument applies over $\C$ as well. Equivalently, because the generators are real, one may apply the real statement to the real and imaginary parts of a polynomial in $\cI_{\cX}^{\C}$ and combine the resulting representations. Applying the scalar statement entrywise proves the matrix version.
\end{proof}

\begin{lemma}[Vanishing ideal]
\label{lem:real-radical}
For $D\ge2$, the ideal $\cI_{\cX}$ defined in \eqref{eq:sphere-ideal} is the full vanishing ideal of $\cX=(\sphere^{D-1})^k$ in $\R[x]$. Consequently:
\begin{enumerate}[label=(\roman*)]
\item Real scalar polynomials that agree on $\cX$ are congruent modulo $\cI_{\cX}$.
\item Complex scalar polynomials that agree on $\cX$ are congruent modulo $\cI_{\cX}^{\C}$.
\item If $F,G\in\C[x]^{s\times s}$ agree on $\cX$, then $F-G\in M_s(\cI_{\cX}^{\C})$, or equivalently $F\equiv G\pmod{\cI_{\cX}^{\C}}$ in the sense of \eqref{eq:matrix-congruence}.
\end{enumerate}
\end{lemma}

\begin{proof}
For a single block, the Fischer decomposition gives
\[
\R[x_j]=\bigoplus_{\ell\ge0}\bigoplus_{r\ge0}\norm{x_j}_2^{2r}\Pharm_\ell(\R^D);
\]
see \cite{DunklXu2014}. Since $\norm{x_j}_2^2\equiv1\pmod{g_j}$, every class modulo $g_j$ has a representative in $\bigoplus_{\ell\ge0}\Pharm_\ell(\R^D)$. Restriction to the sphere is injective on each homogeneous harmonic space, and the restrictions of different degrees are mutually orthogonal.

Choose bases of these harmonic spaces whose restrictions are orthonormal on the sphere. Applying the decomposition blockwise shows that $\R[x]/\cI_{\cX}$ is spanned by products of the chosen basis polynomials, one from each block. Their restrictions are orthonormal in $L^2(\mu_D^{\otimes k})$ and therefore linearly independent. It follows that evaluation on $\cX$ is injective on $\R[x]/\cI_{\cX}$. Thus every real polynomial vanishing on $\cX$ belongs to $\cI_{\cX}$. The reverse inclusion follows because every generator vanishes on $\cX$, proving the vanishing-ideal statement and part~(i).

For part~(ii), write a complex polynomial vanishing on $\cX$ as $p=a+\mathrm{i}b$, with $a,b\in\R[x]$. Both $a$ and $b$ vanish on $\cX$, so part~(i) gives $a,b\in\cI_{\cX}$ and hence $p\in\cI_{\cX}^{\C}$. Applying this argument entrywise to $F-G$ proves part~(iii).
\end{proof}

For related Positivstellens\"atze on odd-dimensional spheres, see \cite{DAngeloPutinar2009}.

\section{Continuity and measurable selection}
\label{app:selection}

This appendix proves \cref{lem:measurable-selector}. We first record an observation that will be used for both selection arguments. Let $X$ be a metric space, let $K$ be a compact metric space, and let $\Gamma:X\rightrightarrows K$ have nonempty compact values and closed graph. Then $\Gamma$ is weakly Borel measurable.

Indeed, for every closed set $C\subseteq K$, the set $\{x\in X:\Gamma(x)\cap C\ne\varnothing\}$ is closed by compactness of $C$ and the closed-graph property. For an open set $U\subsetneq K$, define
\[
C_m:=\{y\in K:\operatorname{dist}(y,K\setminus U)\ge1/m\},\qquad m\ge1.
\]
Each $C_m$ is compact and $U=\bigcup_{m\ge1}C_m$. Hence
\[
\{x\in X:\Gamma(x)\cap U\ne\varnothing\}
=\bigcup_{m\ge1}\{x\in X:\Gamma(x)\cap C_m\ne\varnothing\}
\]
is Borel. The case $U=K$ follows from nonemptiness of the values of $\Gamma$.

\begin{proof}[Proof of \cref{lem:measurable-selector}]
The POVM set $\{(M_1,\ldots,M_k)\in\Herm(B)^k:M_i\succeq0,\ \sum_iM_i=I_B\}$ is nonempty and compact. After substituting $\omega_i=\sigma_i(x)$, the objective in \eqref{eq:discrimination-primal} is jointly continuous in $x$ and the POVM, because each $\sigma_i$ is the restriction of an ambient polynomial. Berge's maximum theorem \cite{AliprantisBorder2006}, applied to this constant feasible-set correspondence, therefore implies that $s$ is continuous.

For the measurable selector, define
\[
\mathsf K:=\{Y\in\Herm(B):0\preceq Y\preceq I_B\},\qquad
\mathsf F(x):=\{Y\in\mathsf K:Y\succeq\sigma_i(x)/k\ \text{for all }i\in[k]\}.
\]
The set $\mathsf K$ is a compact metric space. Each $\mathsf F(x)$ is closed in $\mathsf K$ and contains $I_B$, so it is nonempty and compact. The correspondence $\mathsf F$ has closed graph because the fields $\sigma_i$ are continuous and positive semidefinite inequalities are preserved under limits. The observation above therefore shows that $\mathsf F$ is weakly Borel measurable.

Applying the measurable maximum theorem \cite{AliprantisBorder2006} to the objective $(x,Y)\mapsto-\Tr Y$ yields a Borel measurable minimizer $Y_*(x)$ of $\Tr Y$ over $\mathsf F(x)$. By \eqref{eq:discrimination-dual} and the bounds on optimal dual solutions established immediately after that equation, restricting the dual feasible set to $\mathsf K$ does not change its optimal value. Consequently, $\Tr Y_*(x)=s(x)$, and membership in $\mathsf F(x)$ gives the remaining properties in \eqref{eq:Ystar-properties}.
\end{proof}

\medskip
\noindent\textbf{Alternative selection argument.} With $\mathsf F$ and $\mathsf K$ as above, define the optimal-solution correspondence
\[
\cY_*(x):=\{Y\in\mathsf F(x):\Tr Y=s(x)\}.
\]
Dual attainment in \eqref{eq:discrimination-dual}, together with the bounds on optimal dual solutions, implies that $\cY_*(x)$ is nonempty and compact for every $x\in\cX$. Its graph is closed because $s$ and the fields $\sigma_i$ are continuous. Since all its values lie in the fixed compact metric space $\mathsf K$, the preceding observation implies weak Borel measurability. The Kuratowski--Ryll-Nardzewski selection theorem \cite{KuratowskiRyllNardzewski1965} then gives a Borel measurable selector satisfying \eqref{eq:Ystar-properties}.

\section{The degree-two kernel constant}
\label{app:kernel-constant}

This appendix proves \cref{prop:kernel-constant} and gives the finite matrix representation of $\lambda_{D,L}$. Fix $D\ge2$ and $L\ge1$. Let $p_0,\ldots,p_L$ be real orthonormal polynomials in $L^2(\nu_D)$ with $\deg p_r=r$. For $D\ge3$, these are the Gegenbauer polynomials normalized in $L^2(\nu_D)$. For $D=2$, they are $p_0=1$ and $p_r=\sqrt2\,T_r$ for $r\ge1$.

Writing $q=\sum_{r=0}^Le_rp_r$ with $e=(e_0,\ldots,e_L)^\top\in\R^{L+1}$, the normalization \eqref{eq:kernel-normalization} becomes $\norm{e}_2=1$. Define the real symmetric matrix $T_{D,L}\in\R^{(L+1)\times(L+1)}$ by
\begin{equation}
%(T_L)_{rs}:=\int_{-1}^1p_r(t)p_s(t)g_{2,D}(t)\,d\nu_D(t),\qquad 0\le r,s\le L.
(T_{D,L})_{rs} :=\int_{-1}^1p_r(t)p_s(t)g_{2,D}(t)\,d\nu_D(t), \qquad 0\le r,s\le L.
\label{eq:Toeplitz-definition}
\end{equation}
Equation \eqref{eq:funk-hecke-eigenvalue} gives $\lambda_2(q)=e^\top T_{D,L}e$, and hence
\[
\lambda_{D,L}=\max_{\norm{e}_2=1}e^\top T_{D,L}e=\lmax(T_{D,L}).
\]
In particular, the maximum is attained. Since $g_{2,D}$ has degree two, orthogonality implies $(T_{D,L})_{rs}=0$ whenever $\abs{r-s}>2$, so $T_{D,L}$ is pentadiagonal.

\medskip
\noindent\textbf{Upper bound.} By \eqref{eq:g2}, $g_{2,D}(t)\le1$ on $[-1,1]$. Therefore every normalized $q$ satisfies $\lambda_2(q)\le\int_{-1}^1q(t)^2\,d\nu_D(t)=1$, proving $\lambda_{D,L}\le1$.

% \medskip
% \noindent\textbf{Lower bound.} Take $q(t)=\sqrt D\,t$, which has degree at most $L$. If $v\sim\mu_D$, then $t\sim\nu_D$ has the same distribution as $v_1$. Rotational invariance gives $\E[v_1^2]=1/D$, while the fourth-order moment tensor has the form
% \[
% \E[v_iv_jv_\ell v_m]
% =a\bigl(\delta_{ij}\delta_{\ell m}+\delta_{i\ell}\delta_{jm}+\delta_{im}\delta_{j\ell}\bigr).
% \]
% The identity $\E\norm{v}_2^4=1$ gives $a=1/(D(D+2))$. Consequently,
% \[
% \int_{-1}^1t^2\,d\nu_D(t)=\frac1D,\qquad
% \int_{-1}^1t^4\,d\nu_D(t)=\frac{3}{D(D+2)}.
% \]
% Thus $q$ is normalized, and
% \[
% \lambda_2(q)=\frac{D}{D-1}\left(D\int_{-1}^1t^4\,d\nu_D(t)-\int_{-1}^1t^2\,d\nu_D(t)\right)
% =\frac{D}{D-1}\left(\frac{3}{D+2}-\frac1D\right)=\frac{2}{D+2}.
% \]
% This proves $2/(D+2)\le\lambda_{D,L}\le1$. Since $t\mapsto t^{-1}-1$ is decreasing on $(0,1]$, it also gives $0\le\rho_{D,L}\le D/2$, completing \eqref{eq:lambda-basic}.

\medskip
\noindent\textbf{Lower bound.}
For each integer $j\ge0$, define
\[
m_{2j}:=\int_{-1}^1 t^{2j}\,d\nu_D(t)>0.
\]
The measure $\nu_D$ has density proportional to
$(1-t^2)^{(D-3)/2}$ on $(-1,1)$. Integrating the derivative of
$t^{2j+1}(1-t^2)^{(D-1)/2}$ over $[-1,1]$ gives
\[
(D+2j)m_{2j+2}=(2j+1)m_{2j},
\]
since the boundary terms vanish for $D\ge2$. Hence
\[
\frac{m_{2L+2}}{m_{2L}}=\frac{2L+1}{D+2L}.
\]
Take
\[
q(t):=\frac{t^L}{\sqrt{m_{2L}}},
\]
which is a normalized real polynomial of degree $L$. By
\eqref{eq:funk-hecke-eigenvalue} and \eqref{eq:g2},
\[
\lambda_2(q)
=\frac{1}{D-1}
 \left(D\frac{m_{2L+2}}{m_{2L}}-1\right)
=\frac{2L}{D+2L}.
\]
By maximality of $\lambda_{D,L}$ and the upper bound proved above,
\[
\frac{2L}{D+2L}\le\lambda_{D,L}\le1.
\]
Since $t\mapsto t^{-1}-1$ is decreasing on $(0,1]$, this yields
\[
0\le\rho_{D,L}\le\frac{D}{2L},
\]
establishing \eqref{eq:lambda-basic}.

\medskip
\noindent\textbf{Quadratic rate.} For $D\ge3$, the degree-two case of the squared-kernel distortion estimate in \cite{FangFawzi2021} provides absolute constants $c_0'\ge1$ and $c_1'>0$ such that, whenever $L\ge c_0'D$, there exists a normalized real polynomial $\widetilde q$ of degree at most $L$ satisfying
\[
\abs{\lambda_2(\widetilde q)^{-1}-1}\le c_1'\left(\frac DL\right)^2.
\]
Here the kernel eigenvalues agree with \eqref{eq:funk-hecke-eigenvalue}, since the normalized zonal polynomial is $g_{2,D}$ and the integration measure is $\nu_D$. Enlarge $c_0'$ if necessary so that $c_1'/(c_0')^2<1$. The displayed bound then implies $\lambda_2(\widetilde q)>0$: its reciprocal is defined, and a negative eigenvalue would give $\abs{\lambda_2(\widetilde q)^{-1}-1}>1$. By maximality, $\lambda_{D,L}\ge\lambda_2(\widetilde q)$, so
\[
\rho_{D,L}=\lambda_{D,L}^{-1}-1
\le\lambda_2(\widetilde q)^{-1}-1
\le c_1'\left(\frac DL\right)^2.
\]

For completeness, the same rate on the circle follows directly from the Chebyshev basis. The Jacobi-matrix characterization of orthogonal polynomial zeros \cite{FangFawzi2021} gives a normalized polynomial $q$ of degree at most $L$ such that
\[
\int_{-1}^1tq(t)^2\,d\nu_2(t)=\xi_L,\qquad
\xi_L:=\cos\!\left(\frac{\pi}{2(L+1)}\right),
\]
where $\xi_L$ is the largest zero of $T_{L+1}$. Since $q(t)^2\,d\nu_2(t)$ is a probability measure, Cauchy--Schwarz and $g_{2,2}(t)=2t^2-1$ give
\[
\lambda_{2,L}\ge\lambda_2(q)\ge2\xi_L^2-1
=\cos\!\left(\frac{\pi}{L+1}\right).
\]
For $L\ge2$, the cosine is at least $1/2$. Using $1-\cos\theta\le\theta^2/2$, we obtain
\[
\rho_{2,L}\le\frac{1-\cos(\pi/(L+1))}{\cos(\pi/(L+1))}
\le\frac{\pi^2}{(L+1)^2}
\le\frac{\pi^2}{4}\left(\frac2L\right)^2.
\]
%Thus \eqref{eq:FF-rate} holds for all $D\ge2$ with $c_0=c_0'$ and $c_1=\max\{c_1',\pi^2/4\}$.
Set
\[
c_1:=\max\left\{c_1',\frac{\pi^2}{4},\frac{c_0'}{2}\right\}.
\]
Since $c_0'\ge1$, the preceding estimates give
\[
\rho_{D,L}\le c_1\left(\frac DL\right)^2
\]
for every $D\ge2$ and $L\ge c_0'D$. For the remaining levels
$1\le L<c_0'D$, the monomial bound proved above gives
\[
\rho_{D,L}
\le\frac{D}{2L}
=\frac{L}{2D}\left(\frac DL\right)^2
\le\frac{c_0'}{2}\left(\frac DL\right)^2
\le c_1\left(\frac DL\right)^2.
\]
Thus \eqref{eq:FF-rate} holds for every $D\ge2$ and every $L\ge1$,
with an absolute constant $c_1$ independent of $D$ and $L$.
\end{document}